\documentclass[journal]{IEEEtran}

\usepackage[utf8]{inputenc}
\usepackage{amsmath}
\usepackage{amsthm}
\usepackage{mathrsfs}
\usepackage{amssymb}
\usepackage[english]{babel}
\usepackage{graphicx}
\usepackage{float}
\usepackage{color}
\usepackage{eurosym}
\usepackage[breaklinks]{hyperref}
\usepackage{multirow}
\usepackage{algpseudocode}
\usepackage{algorithm}
\usepackage{bbm}
\usepackage{cite}
\usepackage{dblfloatfix}    

\newtheorem{theorem}{Theorem}
\newtheorem{lemma}{Lemma}

\newtheorem{proposition}{Proposition}
\newtheorem{remark}{Remark}
\allowdisplaybreaks

\makeatletter
\newcommand{\StatexIndent}[1][3]{%
  \setlength\@tempdima{\algorithmicindent}%
  \Statex\hskip\dimexpr#1\@tempdima\relax}
\makeatother

\title{Smart Meter Privacy with Renewable Energy and an Energy Storage Device}
\author{Giulio Giaconi, \IEEEmembership{Student Member, IEEE}, Deniz G\"{u}nd\"{u}z, \IEEEmembership{Senior Member, IEEE}, and H. Vincent Poor, \IEEEmembership{Fellow, IEEE}
\thanks{The work of G. Giaconi was supported by the Engineering and Physical Sciences Research Council (EPSRC) of the U.K. under Grant 1507704. This work was supported in part by the EPSRC through the project COPES under Grant 173605884, in part by the European Research Council under Starting Grant BEACON (agreement 677854), and in part by the U.S. National Science Foundation under Grant CMMI-1435778, Grant ECCS-1549881, and Grant ECCS-1647198. This paper was presented in part at the IEEE International Conference on Communications, London, U.K., June 2015 \cite{Giaconi:2015}.

G. Giaconi and D. G\"{u}nd\"{u}z are with the Department of Electrical and Electronic Engineering, Imperial College London, London,  SW7 2AZ, UK (e-mail: \{g.giaconi, d.gunduz\}@imperial.ac.uk).

H. V. Poor is with the Department of Electrical Engineering, Princeton University, Princeton, NJ 08544 USA (e-mail: poor@princeton.edu).

}}

\begin{document}
\maketitle

\begin{abstract}
A \textit{smart meter} (SM) measures a consumer's electricity consumption and reports it automatically to a utility provider (UP) in almost real time. Despite many advantages of SMs, their use also leads to serious concerns about consumer privacy. In this paper, SM privacy is studied by considering the presence of a renewable energy source (RES) and a rechargeable battery (RB), which can be used to partially hide the consumer's energy consumption behavior. Privacy is measured by the \textit{information leakage rate}, which denotes the average mutual information between the user's real energy consumption and the energy requested from the grid, which the SM reads and reports to the UP. The impact of the knowledge of the amount of energy generated by the RES at the UP is also considered. The minimum information leakage rate is characterized as a computable information theoretic single-letter expression in the two extreme cases, that is, when the battery capacity is infinite or zero. Numerical results are presented for the finite battery capacity case to illustrate the potential privacy gains from the existence of an RB. It is shown that, while the information leakage rate decreases with increasing availability of an RES, larger storage capacity is needed to fully exploit the available energy to improve the privacy.
\end{abstract}

\section{Introduction}

The transition from the legacy power distribution network to the new power grid paradigm, the so-called \emph{smart grid} (SG), is rapidly ongoing. An SG provides many advantages for energy generation, transmission, distribution and consumption thanks to the use of information and communication technologies that enable SGs to monitor and control the power network more effectively \cite{Mo:2012}. In addition, an SG eases the integration of renewable energy sources (RESs), which is a fundamental factor in reducing our dependence on fossil fuels and moving on to a low carbon economy. A key feature of an SG is the advanced metering infrastructure, and in particular smart meters (SMs), which record and report the electricity consumption of a household. SMs that are currently being rolled out in the United Kingdom send measurements every $30$ minutes \cite{smartEnergyGB}, whereas those in Texas send every $15$ minutes \cite{SMT}. The frequency of SM measurements is expected to increase drastically in the near future when renewable energy integration increases and the energy market becomes more efficient by incorporating time-of-usage pricing and demand shifting \cite{Segovia:2011}.

The installation of SMs is rapidly advancing worldwide. For example, all European Union countries are required to have 80\% SM adoption by 2020 and 100\% by 2022 \cite{EurDir:2009}. On the other hand, the information that is collected by SMs may be potentially used for other purposes, thereby raising the question of data privacy. By using nonintrusive appliance load monitoring (NILM) techniques, power consumption load profiles can reveal sensitive information, such as the users' habits, presence at home and working hours, potential illnesses or disabilities, equipment being used, and even which TV channel is being watched \cite{Rouf:2012}. First NILM devices were built in the 80s and were already able to detect the activity of some appliances by knowing their power signature \cite{Hart:1992}. Molina-Markham \textit{et al.} \cite{Molina:2010} showed that it is possible to detect users' activity by simply using off-the-shelf clustering and pattern recognition methods, even without any a priori knowledge of the appliances' power signature. The current state of the art is to consider a factorial hidden Markov model to model the total consumption of various household appliances, whose solution is, however, NP hard. To solve this issue, \cite{Shaloudegi:2016} describes a computationally efficient method based on a semidefinite relaxation combined with randomized rounding.

\subsection{Privacy-Aware SM Techniques}

To date, there are two main families of approaches that have been investigated to provide privacy to consumers. The first family includes approaches that process SM data before sending it to the UP, while approaches in the second family aim at modifying the actual user energy demand. Considered within the first family are methods such as \emph{data obfuscation}, \emph{data aggregation} and \emph{data anonymization}. Data obfuscation, i.e., the perturbation of metering data by adding noise, is a classic method, and has been adapted to SGs in \cite{Kim:2011} and \cite{Bohli:2010}. Among these methods, differential privacy \cite{Dwork:2006}, a well-established concept in the data mining literature based on distorting data to protect the privacy of individuals, is applied to SMs in \cite{Backes:2014}. Along these lines, authors in \cite{Sankar:2013TSG} provide a framework that measures the trade-off between altering data (privacy) and sharing them (utility). Data aggregation, proposed in \cite{Bohli:2010}, \cite{Garcia:2010} and \cite{Li:2011}, considers aggregating power measurements over a group of households so that the UP is prevented from knowing individual consumptions. The aggregation can be performed with or without the help of a trusted third party. Data anonymization mainly considers resorting to pseudonyms rather than the real identities, as in \cite{Petrlic:2010} and \cite{Efthymiou:2010SGC}.

The first family of approaches, however, suffer from a further privacy risk. In fact, the energy consumed by a user is provided directly from the grid, which is fully controlled by the distribution system operator (DSO), i.e., the entity that manages the power grid; and hence, the DSO can embed additional sensors to monitor the energy requested by a household or a business, without fully relying on SM readings. Moreover, any attacker, e.g., a thief or an intelligence agency, may decide to install a sensor for directly monitoring a specific household or business. Another disadvantage of data obfuscation methods is the mismatch between the reported values and the real energy consumption. This prevents the DSO from accurately monitoring the grid states and rapidly reacting to outages, energy theft or other problems. To address these problems, the second family of privacy-preserving approaches directly modifies the actual energy consumption profile of the user, called the \textit{input load} rather than simply modifying the data sent to the UP. This can be done, for example, by filtering the energy via an energy storage device, i.e., a rechargeable battery (RB), as in \cite{Kalogridis:2010SGC,Kalogridis:2011,Yang:2012,Varodayan:2011,Tan:2013JSAC,Tan:2017JIFS,Li:2016Arxiv}, or by using an RES, as originally proposed in \cite{Tan:2013JSAC}. If we denote the energy received from the grid as the \textit{output load}, the idea is to physically differentiate the output load with respect to the input load. Different heuristic algorithms have been proposed, such as the best-effort water-filling algorithm in \cite{Kalogridis:2011} that aims at keeping the output load at its most recent value, or the stepping algorithm in \cite{Yang:2012} that quantizes the power demand into a step function. In \cite{Tan:2017JIFS} the problem is solved in the offline setting by taking the energy cost into account, while the online privacy problem is formulated as a Markov decision process in \cite{Li:2016Arxiv}, and solved numerically in general, while a ``single-letter'' expression is provided for an independent and identically distributed (i.i.d.) input load. In \cite{Farokhi:2017TSM} Fisher information is used as a measure of privacy and, by using the Cram\'{e}r-Rao bound, the variance of the estimation error of any unbiased estimator of the household consumption is maximized by minimizing the trace of the Fisher information matrix. When considering also the presence of an RES, a single-letter solution is given for this problem in \cite{Gunduz:2013ICC,Gomez:2013ISIT,Gomez:2015TIFS} under average and peak power constraints on the available RES. In \cite{Chin:2016TSG} model predictive control is adopted to jointly optimize cost and and privacy in the presence of a battery and local energy generation.

In this paper, we adopt the latter approach, and focus on providing privacy by considering the presence of both an RES and an RB. We study privacy from an information theoretic point of view, and, for some scenarios, provide closed-form expressions for the best privacy performance achievable. A similar model, studied in \cite{Gomez:2015TIFS}, imposes only average and peak power constraints on the RES, which can be a microgrid, capable of providing any amount of energy at each time instant. However, the energy produced by an RES at each time instant is typically random, and its statistics depend on the energy source (e.g., solar, wind) and the energy generator specifications. In addition, the finite-capacity battery imposes further limitations on the available energy. Thus, in this paper we study the minimum amount of user's energy consumption information leaked to the UP by taking into account instantaneous power constraints, as initially proposed in \cite{Giaconi:2015}. While the analysis in \cite{Giaconi:2015} is limited to the two extreme scenarios of zero and infinite battery capacity with a discrete-alphabet input load, here we also study the more practical scenario with a finite-capacity storage device, as well as a continuous-alphabet input load.

Following up on \cite{Varodayan:2011}, \cite{Tan:2013JSAC} and \cite{Gomez:2015TIFS}, we model user's energy consumption profile as a randomly generated time series whose statistics are known by the UP, and measure the user's information leakage by the average mutual information between the input and output load vectors, i.e., between the real energy consumption profile of the appliances and the SM readings, which is called the \textit{information leakage rate}. Mutual information between random variables $X$ and $Y$, $I(X;Y)$, is as a measure of dependence between $X$ and $Y$, which is equal to zero if and only if $X$ and $Y$ are independent. We can also interpret mutual information as the reduction in the uncertainty of the UP about the real energy consumption of the appliances, $X^n$, after receiving the SM measurements, $Y^n$. Thus, minimizing mutual information can be interpreted as a way of improving privacy for SM users. Moreover, mutual information as a privacy measure does not depend on the technological implementation of load monitoring algorithms, and therefore, provides statistical privacy guarantees independent of the computational power of the attacker or the particular monitoring algorithm employed. Mutual information as a measure of privacy leakage has also been considered in other domains, see for example \cite{Chatzikokolakis:2008, Kopf:2007, Clark:2002}.

\subsection{Current Home Batteries and Typical Household Input Loads}


\begin{table*}[!t]
\caption{Specifications of some currently available residential batteries.}
\centering
\begin{tabular}{ |c|c|c|c| }
\hline
\textbf{Residential Battery} & \textbf{Capacity (kWh)} & \begin{tabular}[x]{@{}c@{}} \textbf{RB Charging} \\ \textbf{Peak Power (kW)} \end{tabular} & \begin{tabular}[x]{@{}c@{}} \textbf{RB Discharging} \\ \textbf{Peak Power (kW)} \end{tabular}\\
\hline
Sunverge SIS-6848 \cite{Sunverge} & $7.7$, $11.6$, $15.5$, $19.4$ & $6.4$ & $6$\\
\hline
SonnenBatterie eco \cite{sonnen} & $4-16$ & $3-8$ & $3-8$\\
\hline
Tesla Powerwall \cite{tesla} & $13.5$ & $5$ & $5$\\
\hline
LG RESU 48V \cite{LGresu} & $2.9$, $5.9$, $8.8$ & $3$, $4.2$, $5$ &  $3$,  $4.2$, $5$\\
\hline
Panasonic Battery System LJ-SK84A \cite{panasonic} & $8$ & $2$ & $2$\\
\hline
Powervault G200-LI-2/4/6KWH \cite{powervault} & $2$, $4$, $6$ & $0.8$, $1.2$ & $0.7$, $1.4$\\
\hline
Orison  Panel \cite{orison} & $2.2$ & $1.8$ & $1.8$\\
\hline
Simpliphi PHI 3.4 - 48V \cite{simpliphi} & $3.4$ & $1.5$ & $1.5$\\
\hline
\end{tabular}
\label{tab:batteryCapacity}
\end{table*}

In this section we briefly summarize the specifications of residential batteries available in the market and the general statistics of household energy consumption and generation to illustrate the feasibility of privacy-protection through energy management. Table \ref{tab:batteryCapacity} lists the storage capacity and peak power for some of the currently available batteries for residential use. It is noteworthy that the capacities are in the range of few kWh. A typical household's average energy consumption also lies within the same range, as shown in Table \ref{tab:SManalitics}, where we report the distribution of the average user power consumption over different years obtained from various databases, with different time resolutions. From the Dataport database \cite{pecanstreet}  we observe that, independently from the period considered, the average user demand is less than $2$ kWh for $80-90\%$ of the time. Current batteries charged at full capacity would then be able to satisfy the demand for a few hours only.

\begin{table*}[!t]
\caption{Distribution of average household power consumption (resolution refers to the measurement frequency). Values in each column indicate the percentage of time the average consumption falls into the corresponding interval.}
\centering
\resizebox{\textwidth}{!}{
\begin{tabular}{ |c|c|c|c|c|c|c|c|c|c|c| }
\hline
\textbf{Source} & \textbf{Location} & \textbf{Resolution} & \textbf{Time Frame} & \textbf{\# of Houses} & $\mathbf{[0,0.5]}$ \textbf{kW} & $\mathbf{(0.5, 1]}$ \textbf{kW} &  $\mathbf{(1, 2]}$ \textbf{kW} & $\mathbf{(2,3]}$ \textbf{kW} & $\mathbf{(3, 4]}$ \textbf{kW} & $\mathbf{(4,+\infty)}$ \textbf{kW}\\
\hline
\multirow{5}{*}{\cite{pecanstreet}}  & \multirow{5}{*}{Texas}  & \multirow{5}{*}{$60$ mins} & 01/01/2016 - 31/05/2016    & $512$ & $38$ & $30$   & $20$  & $7$   & $3$ & $2$\\
\cline{4-11}
& &  & 01/01/2015 - 31/12/2015    & $703$ & $36$ & $26$   & $20$  & $9$   & $5$ & $4$\\
\cline{4-11}
& & & 01/01/2014 - 31/12/2014    & $720$ & $39$ & $25$   & $20$  & $8$   & $4$ & $4$\\
\cline{4-11}
& & &01/01/2013 - 31/12/2013    & $419$ & $35$ & $25$   & $21$  & $9$   & $5$ & $5$\\
\cline{4-11}
& & &01/01/2012 - 31/12/2012    & $182$ & $31$ & $26$   & $24$  & $10$  & $5$ & $5$\\
\hline
\cite{intertek} & UK & $2$ mins & 01/05/2010 - 31/07/2011       & $251$ & $18$ & $24$   & $47$  & $11$  & $0$ & $0$\\
\hline
\cite{dred} & Netherlands & $1$ sec &05/07/2015 - 05/12/2015           & $1$   & $98$ & $1.8$  & $0.4$ & $0$   & $0$ & $0$\\
\hline
\cite{Lichman:2013} & France & $1$ min & 16/12/2006 - 26/11/2010   & $1$   & $47$ & $9$    & $28$  & $8$   & $4$ & $2$\\
\hline
\end{tabular}
}
\label{tab:SManalitics}
\end{table*}

\begin{table*}[!t]
\caption{Distribution of average power generated by residential photovoltaic systems. Values in each column indicate the percentage of time the average generation falls into the corresponding interval.}
\centering
\resizebox{\textwidth}{!}{
\begin{tabular}{ |c|c|c|c|c|c|c|c|c|c|c|c|c|c|c| }
\hline
\textbf{Source} & \textbf{Location}  & \textbf{Resolution}  & \textbf{Time Frame} & \textbf{\# of Houses} & $\mathbf{0}$ \textbf{kW}& $\mathbf{(0,0.5]}$ \textbf{kW} &  $\mathbf{(0.5,1]}$ \textbf{kW}& $\mathbf{(1,2]}$ \textbf{kW}& $\mathbf{(2,3]}$ \textbf{kW}& $\mathbf{(3,4]}$ \textbf{kW}& $\mathbf{(4,+\infty)}$ \textbf{kW}\\
\hline
\cite{pecanstreet} & Texas & $60$ min &  01/01/2012 - 31/05/2016 & $351$ & $49$ & $17$ & $7$ & $9$ & $7$ & $6$ & $5$\\
\hline
\cite{microgen} & UK & $30$ min &  01/01/2015 - 31/12/2015 & $100$ & $51.7$ & $36.4$ & $9.8$ & $2$ & $0.1$ & $0$ & $0$\\
\hline
\end{tabular}
}
\label{tab:photovoltaic}
\end{table*}

\begin{table*}[!t]
\caption{Specifications of the solar panels studied in \cite{microgen}. The values in each column indicate the percentage of solar panels that satisfy the corresponding property.}
\centering

\begin{tabular}{ |c|c|c|c|c||c|c||c|c|c|c| }
\hline
\multicolumn{5}{|c||}{\textbf{Solar Panel Area ($m^2$)}} & \multicolumn{2}{c||}{\textbf{Solar Panel Cell Type}} & \multicolumn{4}{c|}{\textbf{Nominal Installed Capacity (kWp)}}\\
\hline
$(0,15]$ &  $(15,20]$ & $(20,25]$ & $(25,30]$ & $(30,+\infty)$ & Monocrystalline & Polycrystalline & $(0,2]$ &  $(2,3]$ & $(3,4]$ & $(4,\infty)$ \\
\hline
$5$ & $35$ & $44$ & $15$ & $1$ & $93$ & $7$ & $4$ & $36$ & $59$ & $1$\\
\hline
\end{tabular}

\label{tab:photovoltaicSpecifications}
\end{table*}

In Table \ref{tab:photovoltaic}  we have also included information about the amount of average power generated via a rooftop solar panel. Locations, technology as well as inclinations and sizes of panels vary, as shown in Table \ref{tab:photovoltaicSpecifications}  for one of the databases considered, where kWp denotes the kilowatt peak, i.e., the output power achieved by a panel under full solar radiation. As expected, around $50\%$ of time, i.e., at night, no energy is generated at all, while there are differences in the distribution of the average values for the two databases considered, due to the different areas considered. If we compare these values with those in Table \ref{tab:batteryCapacity}, we can see that the capacities of current batteries are sufficient to store many hours of average solar energy generated by the solar panels most of the time, for which the infinite battery assumption may be an accurate model.

\subsection{Main Contributions}
The main contributions of this paper can be summarized as follows:
\begin{enumerate}
  \item We provide computable closed-form single-letter expressions for the minimum information leakage rate when the battery capacity is zero and infinite. We provide detailed proofs for these results, which have been stated in \cite{Giaconi:2015} without proofs. These two asymptotic performance results can also be considered as upper and lower bounds on the achievable privacy performance for a more practical SM system with a finite-capacity battery.
  \item For these scenarios, we study the information leakage rate also considering the availability of the RES information at the UP, which provides additional side information to the UP.
  \item For a finite-capacity battery scenario, we propose a suboptimal parameterized energy management policy, and optimize the policy parameters using a policy search technique that exploits stochastic gradient descent. We show numerically that the performance of the proposed energy management policy approaches the one with an infinite battery even with a relatively small battery size. This shows the efficacy of the proposed privacy preservation scheme.
  \item We show that the information leakage rate decreases with the rate of the available RES, and that a larger RB is needed to fully exploit the available energy to improve the privacy.
\end{enumerate}

The remainder of the paper is organized as follows. In Section \ref{sec:SystemModel} the system model is introduced. In Section \ref{sec:infinite} an ideal system with an infinite-capacity battery is studied, while in Section \ref{sec:zero} another extreme case with no energy storage is considered. For both scenarios, we also study the case in which the UP knows the realizations of the renewable energy process. In Section \ref{sec:binary} we study the binary scenario, while in Section \ref{sec:finite} we propose achievable schemes for the generic finite battery capacity scenario, and present the corresponding numerical results. In Section \ref{sec:continuous} a continuous input load is considered, while conclusions are drawn in Section \ref{sec:conclusion}.


\subsection{Notation}
Random variables (RVs) are denoted by capital letters $X, Y$, their realizations by lower-case letters $x, y$, and the corresponding alphabets by calligraphic letters $\mathcal{X}, \mathcal{Y}$. The probability distribution of a RV $X$ taking values in $\mathcal{X}$ is denoted by $p_{X}$. For integers $0 < a < b$, $X_{a}^{b}$ denotes the sequence $(X_a,X_{a+1},\ldots,X_b)$, while $X^b \triangleq X_{1}^b$. All logarithms and exponentials are in base $2$, unless specified otherwise.


\section{System Model}\label{sec:SystemModel}

\begin{figure}[!t]
\centering
\includegraphics[width=\columnwidth]{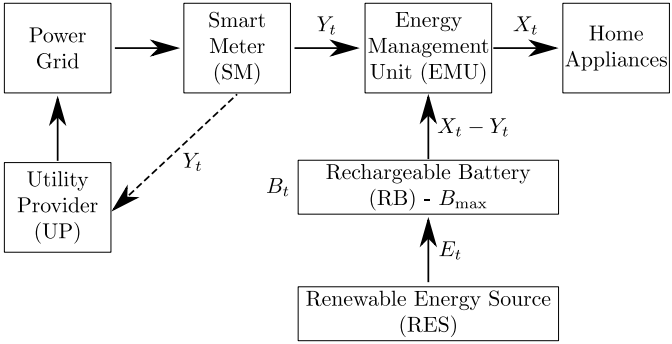}
\caption{System model. $X_t$, $Y_t$, $E_t$ and $B_t$ denote the consumer's energy demand, the SM readings, the energy produced by the RES, and the state of the RB at time $t$, respectively. The dashed line represents the meter readings being reported to the UP.}
\label{fig:SystemModel}
\end{figure}

A discrete time system model is adopted as depicted in Figure \ref{fig:SystemModel}. $X_t \in \mathcal{X}$ is the total amount of power demanded by a user in time slot $t$, where $\mathcal{X}=[0,\ldots,X_{\max}]$, while $Y_t \in \mathcal{Y}$ is the energy received from the UP at time $t$, where $\mathcal{Y}=[0,\ldots,Y_{\max}]$. We call $X_t$ as the \textit{input load} and $Y_t$ as the \textit{output load} to simplify the terminology. For simplicity, we assume that the entries of the input load sequence $\{X_t\}_{t=1}^{\infty}$ are i.i.d. with distribution $p_X$. In time slot $t$, $E_t\in \mathcal{E}$ units of energy are generated from the RES, which becomes available to the energy management unit (EMU) at the beginning of time slot $t$. The entries of the renewable energy sequence $\{E_t\}_{t=1}^{\infty}$ are also i.i.d. with distribution $p_E$ and alphabet $\mathcal{E}=[0,\ldots,E_{\max}]$, while the average renewable energy rate is denoted by $\bar{P}_E \triangleq \mathbbm{E}[E]$. We further consider the presence of an RB in which the renewable energy can be stored for future use. The state of charge (SOC) of the battery at time $t$ is $B_t \in [0,\ldots,B_{\max}]$, and its capacity is $B_{\max}$. We assume no losses in the battery charging and discharging processes.

The EMU always satisfies user's energy demands by drawing energy from either the UP or the RB; that is, outages or demand shifting are not allowed. As a consequence, we have $X_{\max} \geq Y_{\max} \geq X_{\max}-B_{\max}$. We do not allow extra energy to be drawn from the grid and then wasted. This could provide additional privacy, albeit at a significantly higher energy cost. Also, the battery is exclusively for storing the generated renewable energy, and it cannot be recharged with grid energy. While storing grid energy in the battery to be supplied later to the appliances can provide additional privacy \cite{Varodayan:2011}, here we limit the use of the battery to renewable energy storage to isolate and understand the privacy benefits of RESs. Hence, we impose
\begin{equation}\label{eq:constraintY}
0 \leq Y_t \leq X_t, \quad \forall t,
\end{equation}
while $X_t - Y_t$ is the amount of energy obtained from the RB in time slot $t$. The energy retrieved from the battery must be smaller than the energy available in it, i.e.,
\begin{equation}\label{eq:constraintXY}
X_t-Y_t \leq B_t + E_t, \quad \forall t.
\end{equation}

We also consider a peak power constraint $\hat{P}$ on the amount of energy that can be requested at any time from the RB, i.e.,
\begin{equation}\label{eq:constraintPeak}
0 \leq X_t - Y_t \leq \hat{P}, \quad \forall t,
\end{equation}
and for the rest of the paper we assume that $\bar{P}_E \leq \hat{P}$.

Given $(X_t, E_t, B_t)=(x_t, e_t, b_t)$ and the constraints (\ref{eq:constraintY}), (\ref{eq:constraintXY}), and (\ref{eq:constraintPeak}), the set of feasible energy requests at time $t$ is
\begin{equation}\label{eq:feasibleSetY}
\bar{\mathcal{Y}}(x_t,e_t,b_t) \triangleq  \\ \Big\{ y_t \in \mathcal{Y}: [x_t-\min\{b_t+e_t,\hat{P}\}]^+ \leq y_t \leq x_t \Big\},
\end{equation}
where $[a]^+=a$ if $a>0$, and $0$ otherwise.

The battery update equation can be written as
\begin{align}\label{battery_constraint}
	B_{t+1} = \min \Big\{B_{t} + E_t - (X_t - Y_t), B_{\max} \Big\}, \quad \forall t.
\end{align}

We aim at designing \emph{energy management policies} $f=(f_1,f_2,\ldots)$ that decide on the amount of energy to request from the UP at each time $t$, given the previous values of input load $X^t$, renewable energy $E^{t}$, battery SOCs $B^t$, and output load $Y^{t-1}$, i.e.,
\begin{equation*}
f_t: \mathcal{X}^t \times \mathcal{E}^t \times \mathcal{B}^t \times \mathcal{Y}^{t-1} \rightarrow \mathcal{Y}, \quad \forall t,
\end{equation*}
while satisfying (\ref{eq:feasibleSetY}) and (\ref{battery_constraint}), where $f \in \mathcal{F}$ and $\mathcal{F}$ denotes the set of feasible policies, i.e., which produce output load values that satisfy the RB and RES constraints at any time, as well as the battery update equation.

We measure privacy via the \textit{information leakage rate}, defined as the average mutual information rate between the actual user energy consumption and the energy received from the grid, which also corresponds to the reported SM data, i.e.,
\begin{equation}\label{eq:generalForm}
\mathcal{I}_f^{i}(B_{\max},\hat{P}) \triangleq \lim_{n \rightarrow \infty} \frac{1}{n} I \left( X^n; Y^n \right),
\end{equation}
where the subscript $f$ denotes the specific energy management policy employed, and the superscript $i$ stresses the fact that we are considering instantaneous power constraints. Thus, the optimization problem can be written as the minimization of (\ref{eq:generalForm}) over all feasible policies $f \in \mathcal{F}$, i.e.,
\begin{equation}\label{eq:generalMinimization}
\mathcal{I}^{i}(B_{\max},\hat{P}) \triangleq \inf_{f \in \mathcal{F}} \lim_{n \rightarrow \infty} \frac{1}{n} I(X^n;Y^n).
\end{equation}

A single-letter expression for the information leakage rate is provided in \cite{Gunduz:2013ICC,Gomez:2013ISIT,Gomez:2015TIFS} when the EMU is constrained only by the average and peak power constraints. In general, because of the memory effects introduced by the RB and the RES, satisfying the input load from the RB or the RES at some time period may come at the expense of revealing more information about the energy consumption at future time periods. For this reason, the information theoretic analysis typically focuses on the average performance, measured over a period of $n$ time slots, and aims at understanding the fundamental performance bounds by letting this time period go to infinity, i.e., $n \rightarrow \infty$, as in (\ref{eq:generalForm}). However, the definition of the information leakage rate in (\ref{eq:generalForm}) involves $n$-length sequences $X^n$ and $Y^n$, and the asymptotic performance limit corresponds to an infinite-dimensional optimization problem, which cannot be solved numerically. On the contrary, characterizing a single-letter expression allows the optimal solution to be to described as an optimization problem in terms of the single-letter random variables, which can be a finite-dimensional optimization problem when the involved random variables are defined over finite alphabets. Therefore, a single-letter characterization of the information theoretic privacy is desirable to be able to evaluate the minimum possible information leakage rate.

In \cite{Gomez:2013ISIT} the \textit{privacy-power function} $\mathcal{I}(\bar{P},\hat{P})$ is defined  as the minimum information leakage rate that can be achieved when the energy management policy satisfies the average power constraint $\mathbbm{E}\big[\sum_{t=1}^n (X_t - Y_t)\big] \leq \bar{P}$, as well as the peak power constraint $0 \leq X_t - Y_t \leq \hat{P}$, $\forall t$. The privacy-power function has the single-letter characterization provided by the following theorem.
\begin{theorem}\label{th:average_peak}
\cite[Theorem 1]{Gomez:2013ISIT} The privacy-power function $\mathcal{I}(\bar{P},\hat{P})$ for an i.i.d. input load vector $X$ with distribution $p_{X}(x)$ and output load vector $Y$, when the average and peak values of the power provided by the RES are limited by $\bar{P}$ and $\hat{P}$, respectively, is given by
\begin{equation}\label{expr_privacy_power}
\mathcal{I}(\bar{P},\hat{P}) = \inf_{p_{Y|X} \in \mathcal{P}} I\left(X;Y\right),
\end{equation}
where $\mathcal{P} \triangleq \{ p_{Y|X}: y \in \mathcal{Y}, \mathbbm{E}[(X-Y)] \leq \bar{P}, 0 \leq X-Y \leq \hat{P}\}$.
\end{theorem}

\begin{lemma}
\cite[Lemma 1]{Gomez:2013ISIT}
The privacy-power function $\mathcal{I}(\bar{P},\hat{P})$, given above, is a non-increasing convex function of $\bar{P}$ and $\hat{P}$.
\end{lemma}


It is shown in \cite{Gomez:2015TIFS} that, when the input load alphabet is discrete, i.e., $\mathcal{X}=\{0,1,\ldots,X_{\max}\}$, the output load alphabet $\mathcal{Y}$, which is not necessarily discrete, can be restricted to the input load alphabet, i.e., $\mathcal{Y} = \mathcal{X}$, without loss of optimality. Given this restriction and the convexity of the privacy-power function, $\mathcal{I}(\bar{P},\hat{P})$ can be numerically evaluated, e.g., by the efficient Blahut-Arimoto (BA) \cite{Blahut:1972} algorithm. The following lemma states that this property holds also in our setting for the various battery capacities we analyze in the following. Thus, in the discrete case, we can assume that all the involved random processes are defined over finite alphabets and that there is a minimum quantum of energy such that all the aforementioned quantities are integer multiples of this quantum.
\begin{lemma}\label{lemma:discreteAlphabets} If the input alphabet $\mathcal{X}$ is discrete, the output alphabet $\mathcal{Y}$ can be constrained to the input alphabet without loss of optimality.
\end{lemma}

\begin{proof}
The proof is similar to that of \cite[Theorem 2]{Gomez:2015TIFS}. Let $\mathcal{X}$ be the discrete input load alphabet and let $X(y)=\min_{x \in \mathcal{X}}\{x\geq y\}$. Then, for any given energy management policy, and the resultant output load $Y^n$, we define a new output load as $\hat{Y}(t) = X(Y(t))$, that is, $\hat{Y}$ is a post-processed version of $Y$, and $\hat{\mathcal{Y}}=\mathcal{X}$. By construction, we have that $X(t) \geq \hat{Y}(t) \geq Y(t), \forall t$, i.e., the power demanded by the battery cannot have a larger peak value than the original demanded power. Similarly, the new output load satisfies all the instantaneous power constraints as well. This proves that the policy is feasible. Also, the information leakage rate is not increased as $\hat{Y}$ is a deterministic function of $Y$, and thus $X - Y - \hat{Y}$ forms a Markov chain, and $I(X,Y)\geq I(X,\hat{Y})$ by the data processing inequality.
\end{proof}

Here we introduce a generic energy management policy, which we later specialize to the different scenarios we consider. This is a stationary and memoryless policy that generates $Y_t$ randomly using a conditional probability distribution that is based only on the current input load $X_t$ and the available total renewable energy $B_t+E_t$, i.e.,
\begin{equation}\label{eq:generalizedPolicy}
\tilde{p}_{Y|X,B+E} : \mathcal{X} \times (\mathcal{B+E}) \rightarrow \mathcal{Y}.
\end{equation}

Note that, in the presence of an RB, in which the generated renewable energy is stored and used for privacy, a memoryless energy management policy is suboptimal in general, as it ignores the history. However, in the following we show that a memoryless policy is able to achieve the minimum information leakage rate in the two extreme scenarios of $B_{\max}=\infty$ and $B_{\max}=0$.

\begin{figure*}[!b]
\vspace{-3mm}
\hrulefill
\setcounter{equation}{11}
\begin{equation}\label{eq:bestEffort}
    \tilde{p}_{Y|X,B+E}(y|x,b+e)=
\begin{cases}
    p^*_{Y|X}(y|x),      &\text{if } x-y^* \leq b+e \text{ and } y^* \neq x,\\
    p^*_{Y|X}(y|x) + \sum_{ \{y' \in \mathcal{Y}: x-y'>b+e\}} p^*_{Y|X}(y'|x),&\text{if } y^*=x,\\
    0,              &\text{if } x-y^* > b+e.
\end{cases}
\end{equation}
\setcounter{equation}{9}
\end{figure*}


\section{Infinite Battery Capacity}\label{sec:infinite}

In this section we relax the constraint on the battery capacity and consider $B_{\max}=\infty$. This is an extreme situation that may model a battery with a relatively large capacity compared to the average generation rate of renewable energy, $\bar{P}_E$, and the average input load. This scenario provides useful insights on the best achievable privacy performance, and also serves as a bound on the performance achievable with a finite-capacity RB.

In each time slot, the EMU is limited by both the peak power constraint (\ref{eq:constraintPeak}) and the energy available in the RB, which is the difference between the total renewable energy generated and the total energy that has been requested from the battery up to that time, i.e.,
\begin{equation} \label{eq:constrInf}
 	\sum_{t=1}^n (X_t - Y_t) \leq \sum_{t=1}^n E_t, \quad \forall n.
\end{equation}




\subsection{Generated Renewable Energy not Known by the UP}\label{sec:BinfiniteNotKnown}

In this section $E^n$ is treated as a random sequence whose realization is known only to the consumer in a causal manner. This scenario may occur if the renewable energy originates from sources which could be extremely difficult, if not impossible, for the UP to track.

The following theorem states that the minimum information leakage rate when $B_{\max} = \infty$ is equivalent to the average and peak power-constrained scenario, as in \cite{Gomez:2013ISIT}; that is, the cumulative constraints on the EMU policy do not reduce the achievable privacy if the battery capacity is sufficiently large.

\setcounter{equation}{10}
\begin{theorem}\label{th:Binfinite_peak}
If $B_{\max}=\infty$ and the peak power constraint on the amount of energy taken from the RB is $\hat{P}$, then the minimum information leakage rate for an i.i.d. input load $X$ and a renewable energy generation process with average power $\bar{P}_E$, is
\begin{equation}
\mathcal{I}^i(\infty,\hat{P}) = \mathcal{I}(\bar{P}_E, \hat{P}).
\end{equation}
\end{theorem}
\setcounter{equation}{12}

$\mathcal{I}(\bar{P}_E, \hat{P})$ is a trivial lower bound on $\mathcal{I}^i(\infty,\hat{P})$. In the following section an energy management policy that achieves $\mathcal{I}^i(\infty,\hat{P})$ is presented. The proposed policy is a specialization of the generalized memoryless policy introduced in (\ref{eq:generalizedPolicy}).

\subsection{Optimal Energy Management Policy for \texorpdfstring{$B_{\max}=\infty$}{BInfinite}}\label{sec:optimalPolicyBinfinite}

Consider the following energy management policy. In each time slot $t$, the EMU, based on the instantaneous input load $X_t$, decides on the optimal portion of the input load to be received from the grid, $Y^{*}_{t}$, by using the optimal conditional probability distribution $p^*_{Y|X}$ that minimizes (\ref{expr_privacy_power}). If there is enough energy available to fully satisfy the EMU requests, i.e., $B_{t} + E_{t}\geq X_t-Y^{*}_{t}$, the EMU uses $X_t - Y^*_{t}$ units of renewable energy and $Y^*_{t}$ units of energy from the grid, i.e., $Y_t=Y_t^*$; otherwise, all the input load is satisfied directly from the grid, i.e., $Y_t = X_t$, thus leading to the maximum information leakage for that time instant, i.e., the UP learns $X_t$ perfectly. The time instants at which such leakage occurs cannot be computed beforehand, since they depend on the realizations of the renewable energy process, input and output loads. Given the nature of this policy, which tries to follow the optimal policy generated by ignoring the current SOC, we name it the \emph{best-effort energy management policy}. Algorithm \ref{alg:bestEffortPrivacy} summarizes this policy.

Equation (\ref{eq:bestEffort}), shown at the bottom of the page, specializes policy (\ref{eq:generalizedPolicy}) to the best-effort policy. The second case in (\ref{eq:bestEffort}) includes all the instances for which $p^*_{Y|X}$ outputs either $y^*=x$, or an infeasible output, i.e, for which $x-y^*>b+e$.

\begin{algorithm}[t]
\begin{algorithmic}[1]
\Statex
\State{Initial battery SOC: $B_0$.}
\State{Find $p^*_{Y|X}$ that minimizes (\ref{expr_privacy_power}) for given $\bar{P}_E$ and $\hat{P}$.}
\For {$t = 1, \ldots, n$}
\State{Input: $X_t, B_t, E_t$.}
\State{Generate $Y^*_t$ according to $p^*_{Y|X}$.}
\If{$B_t + E_t \geq X_t-Y^*_{t}$}
\State{Optimal policy is followed: $Y_t=Y^*_t$ and $X_t-Y^*_{t}$}
    \StatexIndent[3] taken from the battery.     
\Else
\State{Full leakage occurs: $Y_t=X_t$.}
\EndIf
\State{Next battery state: $B_{t+1}=\min\{B_t+E_t-(X_t-$}
 \StatexIndent[2] $Y_{t}),B_{\max}\}$.
\EndFor
\end{algorithmic}
\caption{Best-Effort Privacy Policy for $B_{\max}=\infty$.}
\label{alg:bestEffortPrivacy}
\end{algorithm}

Since the energy arrival is stochastic, it may seem that very little can be said about the information leakage rate. However, if the condition $\mathbbm{E}[X-Y^*] < \bar{P}_{E}$ holds, then it is possible to show that the number of times full leakage of information occurs due to unavailability of energy is relatively small compared to the operating time of the system. This is proved in the following lemma.

\begin{lemma}\label{lemma:BestEffort}
If $\mathbbm{E}[X-Y^*] < \bar{P}_E$, and the EMU follows the best-effort energy management policy, then almost surely the condition $B_{t}+E_{t}< X_t - Y^*_t$ holds only in finitely many time slots in the limit of infinite horizon.
\end{lemma}

\begin{proof}
Let $\mathbbm{E}[X-Y^*] = \bar{P}_{E} - \epsilon$, for some $\epsilon > 0$. The sequence $E - (X-Y^*) - \epsilon$ has zero mean. By the strong law of large numbers, the sample average of the sequence converges almost surely to its expected value, i.e., the sequence of events $\{\frac{1}{n} \sum_{t=1}^{n} (E_t - (X_t-Y^*_t) - \epsilon) < -\epsilon \}_{n=1}^{\infty}$, and thus the sequence $\{\frac{1}{n} \sum_{t=1}^{n}(E_t- (X_t-Y^*_t))<0\}_{n=1}^{\infty}$ occurs only for finitely many times. This implies that, with $Y^*_t$ generated according to the best-effort policy, the unavailability of energy at any time, $B_t+E_t < X_t - Y^*_t$, occurs only for finitely many times.
\end{proof}


\begin{lemma}\label{lemma:BestEffortPPF}
If $\mathbbm{E}[X-Y^*] < \bar{P}_E$, then the minimum information leakage rate of the best-effort policy tends to $\mathcal{I}^i(\infty,\hat{P})$, as $n \rightarrow \infty$.
\end{lemma}

\begin{proof}
Divide the sequence of input and output loads according to the time instants in which a private SM operation is achieved, i.e., the time instants the EMU can fully emulate $p^*_{Y|X}$, and time instants in which full leakage occurs. From Lemma \ref{lemma:BestEffort} we know that as $n \rightarrow \infty$, there is only a finite number of time instants, say $m$, in which the level of privacy induced by $p^*_{Y|X}$ is not achieved, i.e., for which the condition $B_t+E_t < X_t - Y^*_t$ holds, when $Y^*_t$ is generated based on $p^*_{Y|X}$. We remind that the condition $X_t - Y^*_t < \hat{P}$ always holds. Then, we can write
\begin{subequations}
\begin{align}
&\frac{1}{n} I(X^n;Y^n) = \frac{1}{n} \Big[H(X^{n})-H(X^{n}|Y^{n}) \Big] \label{Binfty_Enotknown_BEpolicy1} \\
 &= \frac{1}{n} \Bigg[ \sum_{t=1}^n H(X_t)-H(X_t|X^{t-1},Y^n )\Bigg] \label{Binfty_Enotknown_BEpolicy2}  \\
&\geq \frac{1}{n} \Bigg[ \sum_{t=1}^{n} H(X_t) - H(X_t|Y_t) \Bigg]  \label{Binfty_Enotknown_BEpolicy3}  \\
&= \frac{1}{n}  \Bigg[ \sum_{t \in \mathcal{T}^C} I(X_t;Y_t=Y^*_t) + \sum_{t \in \mathcal{T}}   I(X_t;Y_t=X_t) \Bigg]   \\ 
&\geq \frac{n-m}{n}\mathcal{I}^i(\infty,\hat{P}) + \frac{m}{n} H(X) \xrightarrow{n \rightarrow \infty} \mathcal{I}^i(\infty,\hat{P}) \label{eq:be3},
\end{align}
\end{subequations}
where $\mathcal{T}$ is the set of instants when full leakage of information takes place, i.e., for which $Y_t=X_t$, and $\mathcal{T}^C$ is the set of time instants in which the output is generated through $p^*_{Y|X}$, i.e., $Y_t=Y^*_t$; (\ref{Binfty_Enotknown_BEpolicy3}) follows since conditioning reduces entropy; (\ref{eq:be3}) follows since $m$ is finite.
\end{proof}

\subsection{Store-and-Hide Energy Management Policy}

Here we provide an alternative energy management policy in the case of an infinite-capacity battery. The \textit{store-and-hide energy management policy} consists of an initial \emph{storage} phase, during which all the energy requests of the user are satisfied from the grid while all the generated renewable energy is stored in the battery, and a second \emph{hiding} phase, during which the EMU deploys the optimal policy $p^*_{Y|X}$.

More formally, consider $n$ time slots. In the first $s(n)$ time slots, the so-called \textit{storage phase}, no privacy is achieved because we have $Y_t=X_t$, for $t=1,2,\ldots,s(n)$. In the remaining $n-s(n)$ time slots, the so-called \textit{hiding phase}, user demand is satisfied by taking energy from both the grid and the battery according to the optimal policy $p^*_{Y|X}$. We assume that $s(n) = o(n)$, with $ \lim_{n \to \infty} s(n)= \infty $, and $\lim_{n \to \infty} n-s(n)= \infty $. The initial waiting time $s(n)$ enables the battery to store on average $s(n) \bar{P}_E$ units of energy. In the following lemma we show that the energy stored in the initial storage phase is sufficient to let the EMU follow the optimal energy management policy $p^*_{Y|X}$ during the hiding phase, without energy outages almost surely. After $s(n)$ units of time, thanks to the energy already stored in the RB, the system is able to overcome the uncertainty in the energy arrival, and is able to adopt the optimal privacy-preserving energy management policy for the remaining time.

\begin{remark}\label{rem:noInfo} It is noteworthy that no information about the recharge process of the battery is required, and all the EMU needs to know is the average power generated by the renewable energy process, $\bar{P}_E$.
\end{remark}

\begin{lemma}\label{lemma:storeAndHide}
With a storage phase of length $s(n) = o(n)$, where $ \lim_{n \to \infty} s(n)= \infty $, and $ \lim_{n \to \infty} n-s(n)= \infty $, the store-and-hide policy satisfies the energy constraints in (\ref{eq:constrInf}) almost surely provided that $\mathbbm{E}[X-Y^*]  < \bar{P}_E$.
\end{lemma}
The proof can be found in Appendix \ref{ap:storeAndHide}.

By means of Lemma \ref{lemma:storeAndHide} it is possible to show that the minimum information leakage rate of the store-and-hide policy approaches $\mathcal{I}^i(\infty,\hat{P})$ as $n \rightarrow \infty$, as shown in the following lemma, whose proof can be found in Appendix \ref{ap:storeAndHidePPF}.

\begin{lemma}\label{lemma:storeAndHidePPF}
If $\mathbbm{E}[X-Y^*]  < \bar{P}_E$, then the information leakage rate of the store-and-hide policy with $s(n)$ as specified in Lemma \ref{lemma:storeAndHide} approaches $\mathcal{I}^i(\infty,\hat{P})$ as $n \rightarrow \infty$.
\end{lemma}

\begin{remark}Even though the two schemes described above achieve the same privacy performance as $n \rightarrow \infty$, they do have some conceptual differences. During the initial phase of energy saving, the store-and-hide policy satisfies all the user demands from the grid leaking full information. Therefore, the SM readings reveal user's activity completely in this period. While the impact of this on the information leakage rate vanishes as $n \rightarrow \infty$, this might not be preferable in practice. Therefore, we believe that the best-effort policy is more appropriate for practical applications.
\end{remark}




\subsection{Generated Renewable Energy Known by the UP}

\begin{figure}[!t]
\centering
\includegraphics[width=\columnwidth]{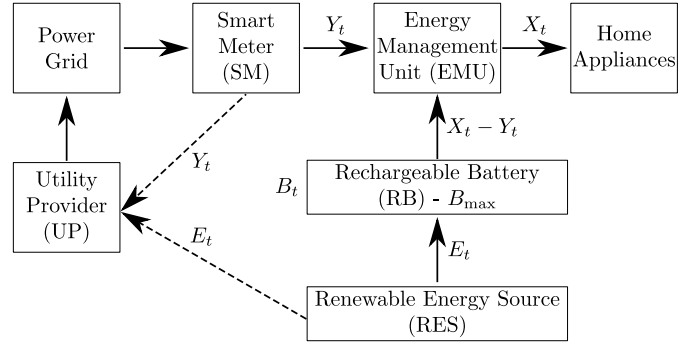}
\caption{The UP has perfect knowledge about the realizations of the renewable energy generation process, in addition to the energy used from the grid that is reported through the SM readings.}
\label{fig:BinfiniteEknownToUP}
\end{figure}

Here we assume that the UP knows the realization of the renewable energy process $E^n$, as highlighted in Figure \ref{fig:BinfiniteEknownToUP}. This scenario can occur if, for example, we consider solar energy as the RES, and the UP can accurately estimate the renewable energy produced from its own observations in nearby locations, weather forecast of the area, and the specifications of the solar panel. This is a worst-case situation and we expect the amount of leaked information in this case to be greater than or equal to that of the previous scenario, in which only the EMU knows the current state of the renewable energy produced. In this setting, the information leakage rate is defined as
\begin{equation}\label{eq:Binfinite_Eknown}
\bar{\mathcal{I}}^{i}(\infty, \hat{P}) \triangleq \inf_{f \in \mathcal{F}} \lim_{n \rightarrow \infty} \frac{1}{n}I(X^n;Y^n|E^n).
\end{equation}

The following theorem states that $E^n$ does not necessarily provide more information to the UP compared to the scenario where the UP does not have access to this information.

\begin{theorem}\label{th:BinfiniteKnown_rate}
If $B_{\max}=\infty$, the minimum information leakage rates for the cases in which $E^n$ is either known or not known to the UP are the same, i.e., $\bar{\mathcal{I}}^{i}(\infty, \hat{P}) = \mathcal{I}^{i}(\infty, \hat{P})$.
\end{theorem}


\begin{proof}
We have the following chain of inequalities:
\begin{subequations}\label{sub}
\begin{align}
\lim_{n \rightarrow \infty} &\frac{1}{n} I(X^n;Y^n|E^n) = \lim_{n \rightarrow \infty} \frac{1}{n} I(X^n;Y^n,E^n) \label{eq:Binfinite_Eknown1}\\
&= \lim_{n \rightarrow \infty} \frac{1}{n} [I(X^n;Y^n) + I(E^n;X^n|Y^n)]\\
&\geq \lim_{n \rightarrow \infty} \frac{1}{n} I(X^n;Y^n), \label{eq:Binfinite_Eknown2}
\end{align}
\end{subequations}
where (\ref{eq:Binfinite_Eknown1}) follows as $X$ and $E$ are independent from each other, and (\ref{eq:Binfinite_Eknown2}) is due to the non negativity of mutual information. Thus, we have $\bar{\mathcal{I}}^{i}(\infty, \hat{P}) \geq \mathcal{I}^{i}(\infty, \hat{P})$.


The inequality in (\ref{eq:Binfinite_Eknown2}) becomes an equality if $I(E^n; X^n|Y^n)=0$. This condition can be achieved by the store-and-hide policy. In fact, at the end of the storage phase the battery is filled up with an infinite amount of energy, and, as a consequence, the optimal policy during the hiding phase $p^*_{Y|X}$ does not need to take the information about the RES into account. This implies that $\lim_{n \rightarrow \infty} I(E^n; X^n|Y^n)=0$; and therefore, $\lim_{n \rightarrow \infty} \frac{1}{n} I(X^n;Y^n|E^n) = \lim_{n \rightarrow \infty} \frac{1}{n} I(X^n;Y^n)$, and that $\bar{\mathcal{I}}^{i}(\infty, \hat{P}) = \mathcal{I}^{i}(\infty, \hat{P})$.
\end{proof}

\section{SM System Without Energy Storage} \label{sec:zero}

In this section we focus on another extreme scenario in which there is no RB for storing extra renewable energy, i.e., $B_{\max}=0$. The renewable energy available at time slot $t$, $E_t$, can be considered as an i.i.d. state information, and could be known, or not, to the UP. Given $E_t$ and $X_t$, the EMU decides on the amount of energy to use from the grid and from the RES. In each time slot $t=1,\ldots,n$ the energy that can be obtained from the RES, $X_t-Y_t$, is limited by the energy generated in time slot $t$, $E_t$, i.e., $0 \leq X_t-Y_t\leq E_t$. Thus, this is an SM system with a stochastic peak power constraint on the energy that the EMU can obtain from the RES. Therefore, this section can be considered as a generalization of \cite{Gomez:2015TIFS}, where the authors consider a fixed peak power constraint.

\begin{remark}
We note that a peak power constraint other than $E_t$ can be easily incorporated to the model, as this would simply correspond to a new instantaneous power constraint of $X_t-Y_t \leq \min\{E_t,\hat{P}\}$. Therefore, for the brevity of the presentation we do not consider a peak power constraint in this section.
\end{remark}


Note that, as opposed to the infinite-capacity battery scenario, here the past has no influence on the energy constraint, since there is no battery, and thus, no memory, in the system.

To analyze this scenario, we first consider the minimum information leakage rate when the generated renewable energy is constant in every time slot, i.e., $\mathcal{E} = \{e\}$, which is known by both the EMU and the UP. The privacy-power function is obtained by considering only a peak power constraint, which can be obtained as a special case of Theorem \ref{th:average_peak}.
\begin{lemma}\label{c:ppfsingle}
If $B_{\max}=0$ and $\mathcal{E} = \{e\}$, the privacy-power function for an i.i.d. input load $X$ is given by $\mathcal{I}(e,e)$.
\end{lemma}

\subsection{Generated Renewable Energy not Known by the UP}\label{sec:BzeroNotKnown}

As in Section \ref{sec:BinfiniteNotKnown}, here the realization of the renewable energy process is assumed to be known only by the EMU, while the UP only knows the probability distribution $p_E$.

\begin{theorem}\label{th:IT}
If $B_{\max}=0$, and the renewable energy produced by the RES is i.i.d. with distribution $p_E$, the optimal information leakage rate, denoted by $\mathcal{I}^i(0)$, is given by
\begin{equation}\label{expr_zero}
\mathcal{I}^i(0) \triangleq \inf_{\substack{p_{Y|X} : p_{Y|X} = \sum_{e \in \mathcal{E}} p_{Y|X,E}(y|x,e);\\ p_{Y|X,E} \in \mathcal{P}^i}} I\left(X;Y\right),
\end{equation}
where $\mathcal{P}^i \triangleq \{p_{Y|X,E}:p_{Y|X,E}(y|x,e)=0 \text{ if } y > x \text{ or } y < x - e \}$.
\end{theorem}

\begin{proof}
\textit{Achievability.} We consider a conditional probability distribution $p_{Y| X, E}(y|x, e)$ that satisfies the conditions of Theorem \ref{th:IT}. At each time instant, for given $x_t$ and $e_t$, $y_t$ is generated independently using the conditional distribution $p_{Y| X,E}(y_t| X_t=x_t, E_t=e_t)$. Since the input and output load sequences are generated i.i.d. with the induced joint distribution $p_X(x)p_{Y|X}(y|x)$, the information leakage rate is given by $I(X;Y)$, whereas the instantaneous peak power constraint is satisfied for all conditional distributions in $\mathcal{P}^i$.

\textit{Converse}. We assume that there is an energy management policy that satisfies the instantaneous peak power constraints, i.e., $x_t - y_t \leq e_t, \forall t$. Then, the information leakage rate satisfies the following chain of inequalities:
\begin{subequations}
\begin{align}
 &\frac{1}{n}  I(X^{n};Y^{n})  =\frac{1}{n}\left[H(X^{n})-H(X^{n}|Y^{n} )\right]  \label{B0_Eknown_Converse1} \\
 &=    \frac{1}{n} \left[ \sum_{t=1}^n H(X_t) - \sum_{t=1}^n H(X_t|X^{t-1},Y^n) \right]  \label{B0_Eknown_Converse2} \\
 &\ge 	\frac{1}{n} \left[ \sum_{t=1}^n H(X_t)-H(X_t|Y_t)\right]  \label{B0_Eknown_Converse3} \\
 &= 	\frac{1}{n} \sum_{t=1}^n I\left(X_t;Y_t\right) \ge \frac{1}{n} \sum_{t=1}^n \mathcal{I}^i(0) = \mathcal{I}^i(0), \label{B0_Eknown_Converse5}
\end{align}
\end{subequations}
where (\ref{B0_Eknown_Converse2}) follows since $X$ is i.i.d.; (\ref{B0_Eknown_Converse3}) follows since conditioning reduces entropy; and (\ref{B0_Eknown_Converse5}) follows from the definition of $\mathcal{I}^i(0)$ in (\ref{expr_zero}).
\end{proof}



\subsection{Generated Renewable Energy Known by the UP}

Here we assume the UP also knows the state $E_t$, $\forall t$.


\begin{theorem}\label{th:ITR}
If $B_{\max}=0$, the input load is i.i.d. with distribution $p_X$, and the amount of generated renewable energy is also known by the UP at each time $t$, then the optimal information leakage rate $\bar{\mathcal{I}}^i(0)$ is given by
\begin{equation}\label{eq:leakage_TR}
\bar{\mathcal{I}}^i(0) = \inf_{p_{Y|X,E} \in \mathcal{P}^i} I\left(X;Y|E \right) = \mathbbm{E}_E [\mathcal{I}(E,E)],
\end{equation}
where $\mathcal{P}^i \triangleq \{p_{Y|X,E}: p_{Y|X,E}(y|x,e)=0  \enspace \textit{if} \enspace y>x \enspace \textit{or} \enspace y<x-e\}$.
\end{theorem}

\begin{IEEEproof} \textit{Achievability} of (\ref{eq:leakage_TR}) follows trivially by employing the optimal $p_{Y|X,E}$ that minimizes (\ref{eq:leakage_TR}) at each time slot. To prove the converse, we show that any energy management policy that satisfies the stochastic peak power constraint at each time instant satisfies the following chain of inequalities:
\begin{subequations}
\begin{align}
 \frac{1}{n} &I(X^{n};Y^{n}|E^n) \nonumber \\
 &= \frac{1}{n} \left[H(X^{n}|E^n)-H(X^{n}|Y^{n},E^n ) \right]\\  
 &=    \frac{1}{n} \Bigg[ \sum_{t=1}^n H(X_t|X^{t-1},E^{n})  - H(X_t|X^{t-1},Y^n,E^n) \Bigg]  \\ 
 &\ge 	\frac{1}{n} \left[ \sum_{t=1}^n H(X_t|E_t)-H(X_t|Y_t,E_t)\right]  \label{B0_Eknown_converse3} \\
 &= 	\frac{1}{n} \sum_{t=1}^n  \sum_{k=1}^{|\mathcal{E}|} p_{E}(E=e_k)   I\left(X_t;Y_t|E_t=e_k\right)  \label{B0_Eknown_converse4} \\
 &\ge 	\frac{1}{n} \sum_{t=1}^n  \sum_{k=1}^{|\mathcal{E}|} p_{E}(E=e_k) \mathcal{I}\left(e_k,e_k\right)  \label{B0_Eknown_converse5} \\
 &= \sum_{k=1}^{|\mathcal{E}|} p_{E}(E=e_k) \mathcal{I}\left(e_k,e_k\right) = \mathbbm{E}_E \left[ \mathcal{I}(E,E) \right],  \label{B0_Eknown_converse6}
\end{align}
\end{subequations}
where (\ref{B0_Eknown_converse3}) follows because $X$ and $E$ are independent of each other and across time, and conditioning reduces entropy; (\ref{B0_Eknown_converse4}) follows by explicitly considering all the states of $E_t$; and (\ref{B0_Eknown_converse5}) follows from Lemma \ref{c:ppfsingle}.
\end{IEEEproof}





From the chain rule of mutual information, we have
\begin{subequations}
\begin{align}
 I(X;Y,E)  &= I(X;E) + I(X;Y|E) = I(X;Y|E),  \label{B0_Eknown_chain1}\\
 I(X;Y,E)  &= I(X;Y) + I(X;E|Y)  \label{B0_Eknown_chain2},
\end{align}
\end{subequations}
where (\ref{B0_Eknown_chain1}) follows since $X$ and $E$ are independent of each other. From (\ref{B0_Eknown_chain1}) and (\ref{B0_Eknown_chain2}), we get $I(X;Y) \leq I(X;Y|E)$. Hence, from Theorems \ref{th:IT} and \ref{th:ITR}, we have $\mathcal{I}^{i}(0) \leq \bar{\mathcal{I}}^{i}(0)$, as expected.


\section{Binary Scenario} \label{sec:binary}

In order to provide further insights into the behavior of the information leakage rate, here we consider a simple scenario with binary energy demands, binary energy generation and binary output load, i.e., $\mathcal{X}= \mathcal{E} = \mathcal{Y}=\{0,1\}$. This scenario may represent appliances that are either on or off/standby. $X$ and $E$ follow independent Bernoulli distributions with $\Pr\{X=1\}=q_x$ and $\Pr\{E=1\} = p_e$, respectively. We compare the minimum information leakage rates for the infinite and zero battery scenarios.

If $B_{\max}= \infty$, the minimum information leakage rate can be characterized explicitly as
\begin{multline}\label{eq:binaryBinfinite}
\mathcal{I}^{i}(\infty,1) = \mathcal{I}(p_e, 1) =\\
\begin{cases}
    p_e \log p_e - q_x \log q_x \\ \quad- (1 - q_x+ p_e) \times \log (1- q_x + p_e), & \text{if } p_e \leq q_x,\\
    0,              & \text{otherwise},
\end{cases}
\end{multline}
where we set the peak power constraint to $\hat{P}=1$.

When $B_{\max}=0$, there are two scenarios. If the generated renewable energy is known only by the EMU, the minimum information leakage rate for this scenario is given by
\begin{equation}
\mathcal{I}^{i}(0;p_e,p_v,q_x) = h(1 - q_x + q_x p_e p_v ) - q_x  h(p_e p_v) \label{B0_Ekonwn_binary2} ,
\end{equation}
where $h(\cdot)$ is the binary entropy function defined as $h(p) \triangleq - p \log p - (1-p) \log(1-p)$, $q_x$ is fixed, and $p_v$ is the probability of using the energy available in the battery whenever $X=1$ and $E=1$.
\begin{proposition}
For every $p_e$ and $q_x$, the information leakage rate $\mathcal{I}^{i}(0;p_e,p_v,q_x)$ is minimized with $p_v=1$.
\end{proposition}

\begin{IEEEproof}
The proof follows from observing that $\frac{\mathrm{d}\mathcal{I}^i(0,p_v)}{\mathrm{d}p_v} \leq 0, \forall p_e, q_x$. Thus, the minimum of $\mathcal{I}^{i}(0;p_e,p_v,q_x)$ is reached when $p_v$ takes its maximum value, i.e., $p_v=1$.
\end{IEEEproof}


When $E_t$ is known also by the UP, if the peak power constraint is $e=1$, no information is leaked, whereas if $e=0$, the input load is known perfectly by the UP, leading to a leakage of $H(X)$. Hence, the minimum information leakage rate when the state information is known by the UP is
\begin{equation}
\bar{\mathcal{I}}^i(0;p_e,q_x)=(1-p_e) h(q_x).
\end{equation}

Numerical comparison of the information leakage rate for zero and infinite battery capacities in the binary scenario will be presented in the next section together with the results corresponding to a finite battery capacity.

\section{Finite Battery Capacity} \label{sec:finite}

A closed-form expression for the finite-capacity battery scenario is elusive as the presence of a finite battery brings memory into the system, and the future energy usage depends on how much renewable energy has been generated in the previous time slots, how much of that energy has already been used by the EMU, and how much is available in the RB. Instead, we propose a low-complexity energy management policy and compare it to the two previous scenarios, which represent upper and lower bounds on the system performance for the finite battery scenario.



\subsection{Binary Alphabet: \texorpdfstring{$\mathcal{X}=\mathcal{E}=\mathcal{Y}=\{0,1\}$}{Binary}}\label{sec:binaryFinite}

In this setting $X$, $E$ and $Y$ have binary alphabets and we consider a discrete-time system, modeled via a finite state machine. As in Section \ref{sec:binary},  we set $\Pr\{X=1\} = q_x$ and $\Pr\{E=1\} = p_e$, while $V^n \triangleq X^n-Y^n$ represents the energy taken by the EMU from the battery, with $\mathcal{V}=\{0,1\}$.

\begin{figure}[!t]
\centering
\includegraphics[width=1\columnwidth]{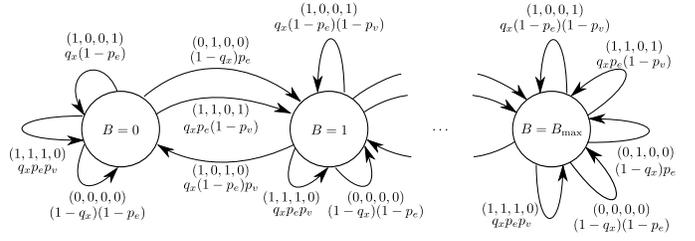}
\caption{Finite state diagram for the evolution of the battery with $\mathcal{B}=\{0,1,\ldots B_{\max}\}$ and $\mathcal{X}=\mathcal{E}=\mathcal{Y}=\{0,1\}$ for the battery-independent policy. The $4$-tuple $(x,e,v,y)$ represent for every time $t$ the values of the input load, the renewable energy produced, the energy taken out of the battery by the EMU, and the output load, respectively.}
\label{fig:unitSizeModelMulti}
\end{figure}


\subsubsection{Battery-independent Policy}\label{sec:batteryInvariantPolicy}

Here we consider a time-invariant policy according to which the evolution of the battery state can be modeled as the Markov chain of Figure \ref{fig:unitSizeModelMulti}, where the $4$-tuples $(x,e,v,y)$ represent the realization at time $t$ of the input load $X$, the renewable energy $E$, the energy taken from the battery by the EMU $V$, and the output load $Y$, respectively. At every time, the RB can be charged, discharged or remain in the current SOC, depending on the transition probabilities. We note that a similar model has been adopted in \cite{Tan:2013JSAC}, with the difference that in \cite{Tan:2013JSAC} the RB can also store energy from the grid. We define $p_{v}$ as the probability that the energy is taken from the battery provided that the user is asking for energy and that there is energy available for use, i.e., $p_v\triangleq\Pr\{V=1\big|X=1,E+B\geq1\}$. Since the value of $p_v$ does not change according to the current battery state, we name this policy \textit{battery-independent policy}. Table \ref{tab:summaryTransFinite} lists all the possible states and transition probabilities for this scenario. In particular, the table shows for each transition from $B_t$ to $B_{t+1}$ and each combination of the tuple $(X_t,E_t,V_t,Y_t)$ the corresponding transition probability.

To compute the information leakage rate, all the distributions are considered to be Bernoulli. For $B_{\max}=\infty$ and $B_{\max}=0$ we use the single-letter expressions derived in Section \ref{sec:binary}, and set $\hat{P}=1$ for $B_{\max}=\infty$. For a finite-capacity battery, we implement the achievable scheme described above, and by means of the algorithm in \cite{Arnold:2006} we simulate the system for very long sequences and evaluate the information leakage between the input and the output loads numerically and for different battery capacities. Moreover, for each $p_e$, we find the value of $p_v$ that achieves the minimum information leakage rate by searching over a discretized set of $p_v$ values. As an example, Figure \ref{fig:B_1optimalPVBinary} represents the optimal $p_v$ values for each $p_e$, when the input load is uniformly distributed and $B_{\max}=\{1,2,5,10\}$. In the figure, $p_e=0$ is not represented because, regardless of $p_v$, the leakage when $p_e=0$ is always equal to the entropy of the input load. Also, the figure shows that for higher $p_e$ values, the minimum leakage is achieved for $p_v=1$, i.e., it is better to always use the energy when available.


\begin{table}[!t]
\caption{Tuples and transition probabilities for the battery-independent policy when $\mathcal{X}=\mathcal{E}=\mathcal{Y}=\{0,1\}$.}
\centering
\resizebox{\columnwidth}{!}{
\begin{tabular}{ |c|c|c|c|c|c|c| }
\hline
$\mathbf{B_t}$ & $\mathbf{X_t}$ & $\mathbf{E_t}$ & $\mathbf{V_t}$ & $\mathbf{Y_t}$ & $\mathbf{B_{t+1}}$ & \textbf{Transition Probability}\\
\hline
\multirow{5}{*}{$B_t=0$} & $0$ & $0$ & $0$ & $0$ & $0$ & $(1-q_x)(1-p_e)$\\
\cline{2-7}
& $0$ & $1$ & $0$ & $0$ & $1$ & $(1-q_x)p_e$ \\
\cline{2-7}
& $1$ & $0$ & $0$ & $1$ & $0$ & $q_x(1-p_e)$ \\
\cline{2-7}
& $1$ & $1$ & $0$ & $1$ & $1$ & $q_x p_e (1-p_v)$ \\
\cline{2-7}
& $1$ & $1$ & $1$ & $0$ & $0$ & $q_x p_e p_v$ \\
\hline
\hline
\multirow{6}{*}{$0< B_t \leq B_{\max}$} & $0$ & $0$ & $0$ & $0$ & $B_t$ & $(1-q_x)(1-p_e)$ \\
\cline{2-7}
& $0$ & $1$ & $0$ & $0$ & $\min\{B_t+1,B_{\max}\}$ & $(1-q_x)p_e$\\
\cline{2-7}
& $1$ & $0$ & $0$ & $1$ & $B_t$ & $q_x(1-p_e)(1-p_v)$ \\
\cline{2-7}
& $1$ & $0$ & $1$ & $0$ & $B_t-1$ & $q_x (1-p_e) p_v$ \\
\cline{2-7}
& $1$ & $1$ & $0$ & $1$ & $\min\{B_t+1,B_{\max}\}$ & $q_x p_e (1-p_v)$ \\
\cline{2-7}
& $1$ & $1$ & $1$ & $0$ & $B_t$ & $q_x p_e p_v$ \\
\hline
\end{tabular}}
\label{tab:summaryTransFinite}
\end{table}

\begin{figure}[!t]
\centering
\includegraphics[width=0.8\columnwidth]{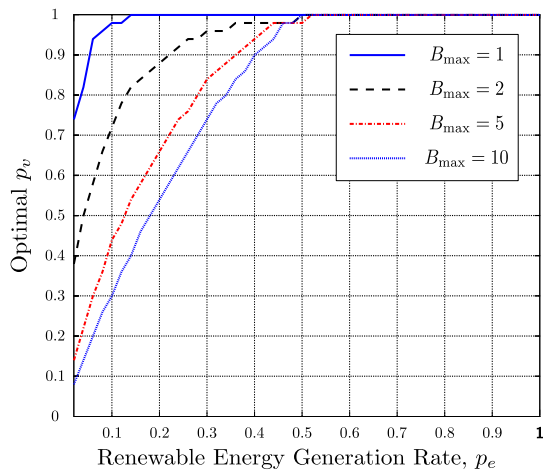}
\caption{Optimal $p_v$ for the binary scenario and various battery capacities, when $q_x = 0.5$.}
\label{fig:B_1optimalPVBinary}
\end{figure}


\subsubsection{Battery-conditioned Policy}\label{sec:batteryConditioned}

Here we consider a policy, in which $p_v$, as defined before, can differ for different battery SOCs, i.e., the policy is characterized by a specific $p_{v_i}$ for each battery SOC $B_t=i$, for $i=\{0,\ldots,B_{\max}\}$. Thus, we now have the vector
\begin{equation}
\bar{p}_v=[p_{v_0},p_{v_1}, \ldots, p_{v_{B_{\max}}}].
\end{equation}

To find the optimal $\bar{p}_v$ for each $p_e$ and $B_{\max}$ we deploy a stochastic gradient descent algorithm, specifically we use the least square-based finite difference method to approximate the gradient \cite{Deisenroth:2013}. Briefly, the algorithm works as follows. At any step, small perturbations are applied to each $p_{v_i}$ according to a uniform distribution over a predefined interval, and the leakage corresponding to the resulting perturbed vector $\bar{p}_v$ is computed. The gradient of the leakage function can thus be approximated numerically by employing the leakage corresponding to a number of different perturbations. A new $\bar{p}_v$ is finally computed using the gradient estimate and a predefined learning rate, and its corresponding leakage is determined and compared with that of the previous step. If the difference between the two leakage rates is below a certain threshold, the algorithm stops. Otherwise, the algorithm keeps on iterating.

Figure \ref{fig:comparison} shows the information leakage rate with respect to the renewable energy generation rate $p_e$, for different battery capacities. For $B_{\max}=\{1,2,5,10\}$, we adopt the battery-conditioned policy, which has only a small gain with respect to the battery-independent policy. In particular, this gain is focused around smaller $p_e$ values. As expected, the least information leakage rate is achieved when $B_{\max}=\infty$ and $\hat{P}=1$, while the maximum leakage occurs when $B_{\max}=0$ and the UP knows the renewable energy process realizations. When $B_{\max}=0$ the information leakage rate reduces significantly if the state is not known by the UP and, more interestingly, we observe that the performance of the proposed suboptimal memoryless scheme approaches that of the infinite-capacity battery with relatively small battery sizes. In addition, we can see that the gain from the battery is much higher when the renewable generation rate is higher, i.e., when $p_e$ is high. This is expected because when $p_e$ is low, there is less energy to be stored for future time slots.

\begin{figure}[!t]
\centering
\includegraphics[width=1\columnwidth]{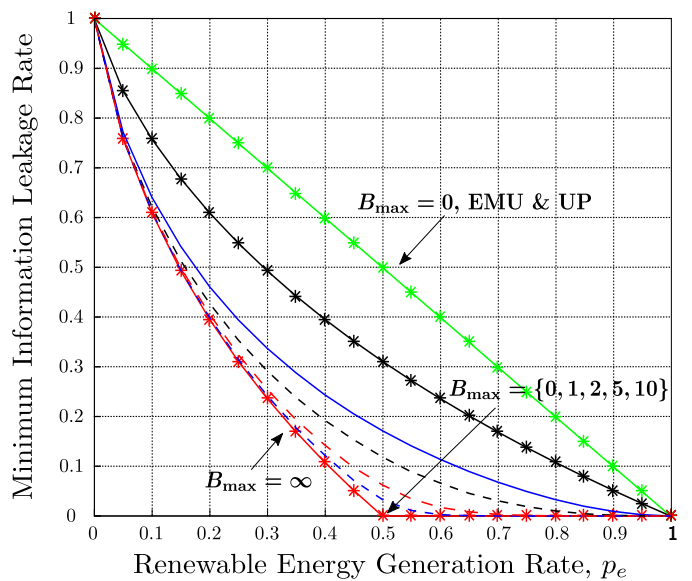}
\caption{Minimum information leakage with respect to the renewable energy generation rate $p_e$ with $\mathcal{X}=\mathcal{E}=\mathcal{Y}=\{0,1\}$ for the battery-conditioned policy. As $B_{\max}$ increases, the performance rapidly approaches that of $B_{\max}=\infty$ and $\hat{P}=1$.}
\label{fig:comparison}
\end{figure}

\subsection{Larger Alphabets: \texorpdfstring{$|\mathcal{X}|=|\mathcal{Y}|=|\mathcal{E}|>2$}{GeneralCase}}

Here we consider larger alphabets for $X$, $E$ and $Y$. As the alphabet sizes grow, so does the complexity of searching for the optimal policy. Instead, we consider the following suboptimal policy. At each time instant, the policy chooses among using all of the available energy, half of it, or no energy at all and we model the probability $p_v$ as in the following:
\begin{equation}\label{eq:transProbModerate}
p_v(B_t+E_t,X_t)=
\begin{cases}
(p_1,p_4),  &\text{if } B_t + E_t < X_t, \\
(p_2,p_5),   & \text{if } B_t + E_t = X_t, \\
(p_3,p_6),    & \text{if } B_t + E_t > X_t.
\end{cases}
\end{equation}

The probability pairs in (\ref{eq:transProbModerate}) refer to the probability of using all the available energy and the probability of using half of it. Therefore, we have $0 \leq p_i \leq 1$, for $i=1,\ldots, 6$, and $p_i+p_{i+3}\leq 1$, for $i=1,2,3$. For example, if $B_t + E_t < X_t$, all of the available energy is used with probability $p_1$, half of it, or the nearest integer value lower than that, is used with probability $p_4$, and none of it is used with probability $1-p_1-p_4$. 

Figure \ref{fig:comparisonFinite} shows the results for the scenario for $|\mathcal{X}|=|\mathcal{E}|=|\mathcal{Y}|=5$ when $B_{\max}=\{0,1,2,\infty\}$. The input load is uniformly distributed over the alphabet $\mathcal{X}$, while the renewable energy generation follows a binomial distribution with parameters $|\mathcal{X}|$ and $p_e$. 
The information leakage rate for the infinite and zero battery scenarios is computed by using the single-letter expressions which are evaluated by efficient numerical algorithms, specifically the BA algorithm \cite{Blahut:1972} and the CVX package \cite{cvx}. In particular, for $B_{\max}=\infty$ we set $\hat{P}=X_{\max}$. For the finite battery scenario, we adopt the aforementioned policy and optimize the performance by trying different combinations of the probabilities $p_i$, $1 \leq i \leq 6$. Similar considerations to that of Figure \ref{fig:comparison} can be drawn for Figure \ref{fig:comparisonFinite} as well.

\begin{figure}[!t]
	\centering
	\includegraphics[width=1\columnwidth]{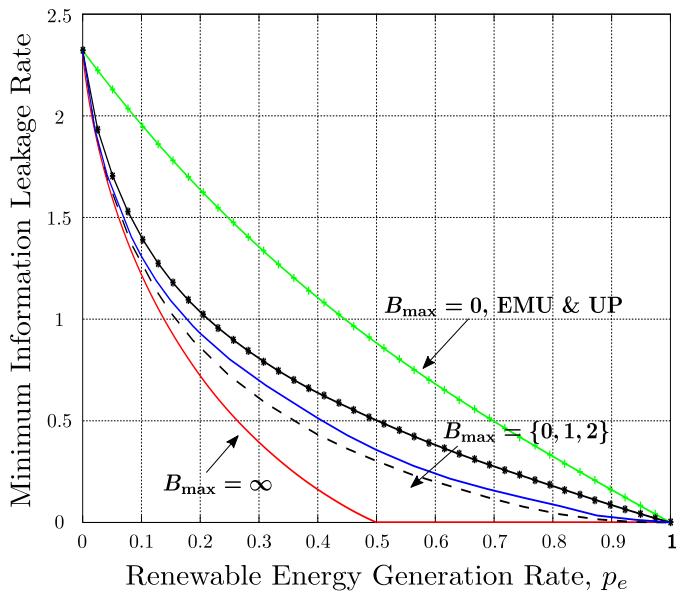}
	\caption{Minimum information leakage with respect to the renewable energy generation rate $p_e$ with $\mathcal{X}=\mathcal{E}=\mathcal{Y}=\{0,1,2,3,4\}$. The leakage for $B_{\max}=\infty$ has been found by setting $\hat{P}=4$.}
	\label{fig:comparisonFinite}
\end{figure}

\begin{remark} We remark here that, in order to isolate the privacy benefits of RESs, we do not allow charging the battery directly from the grid, which can potentially reduce the information leakage. It is known that modulating grid energy intake by employing a storage device provides privacy even in the absence of an RES \cite{Varodayan:2011,Li:2016Arxiv}, or jointly with an RES \cite{Giaconi:2016}. The additional privacy benefits  of allowing charging of the RB from the grid will depend on the battery capacity. When $B_{\max}=\infty$, perfect privacy can be achieved by charging the battery initially, and using the battery throughout the operation. In the other extreme scenario, that is, when $B_{\max}=0$, obviously it is not possible to charge a non-existent battery from the grid. We leave a more detailed study of a finite-capacity storage device that can be charged by both the RES and the grid as a future work. 
\end{remark}

\section{Continuous Input Loads} \label{sec:continuous}

In the simulation results presented above, we have considered discrete alphabets for all the involved random variables. A set of fixed discrete values for the energy demands may not be an accurate model for all the appliances in the real world. However, as discussed in Section \ref{sec:SystemModel}, such hypothesis enables to constrain the output alphabet to the input alphabet without loss of optimality and to apply efficient algorithms to find the minimum amount of information leakage.

For continuous input loads, the optimal alphabet is also continuous. Thus, low-complexity numerical algorithms, such as the BA algorithm, cannot be applied. However, one can provide a lower bound on the privacy-power function by using the Shannon lower bound (SLB) \cite{Cover:1991,Berger:1971}, which has been introduced by Shannon, and widely used in the literature to provide a computable lower bound to the rate-distortion function. Although it is not always a tight bound, it is shown in \cite{Gomez:2015TIFS} that the SLB provides a tight bound for the information leakage rate for an exponentially distributed input load.
The SLB for the rate distortion function $R(D)$ is defined as $H(X)-\phi(D)$ where $\phi(D)= \max_{\substack{p: \sum_{i=1}^{m} p_i d_i \leq D}} H(p)$. The truncated exponential distribution maximises the entropy for a given mean value $\bar{P}$ and a peak power constraint $0 \leq X \leq \hat{P}$ \cite{Cover:1991} and has the form \cite{Gomez:2013ISIT}
\begin{equation}
    f_X(x)=
\begin{cases}
    \frac{1}{\lambda_0} e^{-\frac{x}{\lambda_1}},& \text{if } 0 \leq x \leq \hat{P},\\
    0,              & \text{otherwise},
\end{cases}
\end{equation}
where $\lambda_0 \geq 0$ and $\lambda_1 \geq 0$ are chosen to satisfy the constraints on the moments.
Thus, the SLB for the privacy-power function introduced in Theorem \ref{th:average_peak} is given by
\begin{equation}
\mathcal{I}_{SLB}(\bar{P},\hat{P}) = h(X) - \frac{1}{\ln{2}}\bigg( \log(\lambda_0) + \frac{\bar{P}}{\lambda_1}\bigg).
\end{equation}

Authors in \cite{Gomez:2013ISIT} show that the SLB is indeed achievable for peak and average power constraints, by finding the conditional distribution $f_{Y|X}(y|x)$ that satisfies the SLB with equality, provided that the energy coming from the battery $X-Y$ is distributed according to a truncated exponential distribution with mean $\bar{P}$ and peak $\hat{P}$.

Authors in \cite{Gomez:2015TIFS} provide the SLB for the average power constraint, which, as we have shown, is equivalent to the infinite-capacity battery scenario.



\subsection{No Battery -  Renewable Energy not Known by the UP}
Here only a peak power constraint is considered, i.e., $X-Y$ is constrained by $0 \leq X-Y \leq \hat{P}$. The distribution that maximises the entropy over an interval is the uniform distribution
\begin{align} \label{eq:SLBzero}
    f_X(x)=
\begin{cases}
    \frac{1}{\hat{P}}, & \text{if } 0 \leq x \leq \hat{P},\\
    0,              & \text{otherwise}.
\end{cases}
\end{align}

For a fixed $\hat{P}$, the differential entropy of this distribution is $\log(\hat{P})$. Then, the SLB in the case of zero capacity battery is
\begin{equation}
\mathcal{I}_{SLB}(\hat{P}) = h(X) - \log(\hat{P}),
\end{equation}
where $\hat{P}$ is a RV with a certain known distribution.

\subsection{No Battery - Renewable Energy Known by the UP}
As in the previous scenario, only peak power constraints are considered and thus the entropy maximising distribution is still the uniform distribution (\ref{eq:SLBzero}). The privacy-power function is given by the expected value over the distribution of the states of the privacy-power function related to every state. Hence, the SLB is
\begin{equation}
\mathcal{I}_{SLB}(\hat{P}) = \mathbbm{E}_{\hat{P}}[h(X) - \log(\hat{P})].
\end{equation}

\section{Conclusions} \label{sec:conclusion}

We have studied information leakage in an SM system by considering an RES along with an RB. For infinite and zero battery capacities, we have provided single-letter information theoretic expressions for the minimum information leakage rate, which can be efficiently evaluated when the input load has a discrete alphabet. For these scenarios, we have also studied the information leakage rate when the UP knows the exact amount of renewable energy generated in each time slot. In addition, for the finite-capacity battery scenario, we have proposed a suboptimal low-complexity energy management policy, and evaluated the corresponding privacy performance using a stochastic gradient descent algorithm. Our results show that the privacy achieved by the proposed low-complexity policy approaches the theoretical lower bound obtained by assuming an infinite-capacity battery with a relatively small battery capacity, especially when the generation rate of the RES is low or high.


\begin{appendices}

\section{Proof of Lemma \ref{lemma:storeAndHide}}\label{ap:storeAndHide} 

\begin{proof}
During the hiding phase, the random variable $Q=E-X+Y^*$ is i.i.d., as $E$ and $X$ are i.i.d and $Y^*$ is generated from $X$ through a memoryless policy. $Q$ can assume both positive and negative values with positive probability. The stochastic process
\begin{equation}
S_t = Q_1 + Q_2 + \ldots Q_t, \quad \forall t,
\end{equation}
is a random walk based on $Q$ that moves along the battery SOC axis. Since by hypothesis $\mathbbm{E}[E]=\bar{P}_E > \mathbbm{E}[X-Y^*]$, then $\mathbbm{E}[Q]=\mathbbm{E}[E-X+Y^*]>0$, meaning that the random walk $S_t$ has a positive drift, i.e., as $t \rightarrow \infty$, $S_t$ drift towards the positive values of the SOC axis.

By the law of large numbers, when $s(n) \rightarrow \infty$ the amount of energy stored in the battery at the end of the \emph{storage phase} is $s(n)\bar{P}_E$, almost surely. Let $\alpha\triangleq -s(n)\bar{P}_E$. When $s(n) \rightarrow \infty$, $\alpha \rightarrow -\infty$. At $s(n)+1$, when the \emph{hiding phase} begins, the energy in the battery is used according to the optimal privacy-preserving policy $p^*_{Y|X}$ and the random walk state is $S_1=Q_1=E_1-X_1+Y^*_1$. For any $t$, $s(n)\bar{P}_E+S_t$ represents the battery SOC at time $t$. Our objective is to prove that the battery is never emptied, i.e., that the probability of crossing the threshold $\alpha$ for any time $t$ is zero:
\begin{equation}
\Pr\{S_t \leq \alpha\}=0, \quad \forall t.
\end{equation}

\begin{figure}[!t]
	\centering
	\includegraphics[width=1\columnwidth]{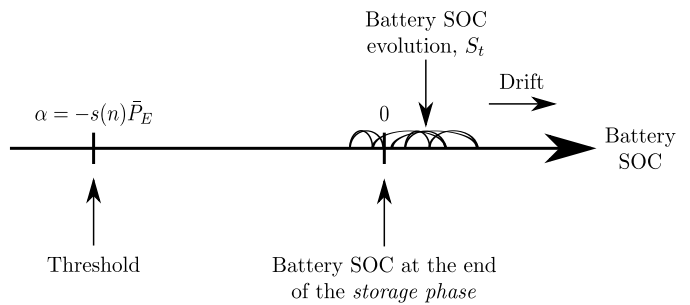}
	\caption{The battery SOC evolution is represented by a random walk that starts at the beginning of the hiding phase and has a drift towards the positive direction of the battery SOC axis. We want to guarantee that the threshold $\alpha$ is never crossed by the random walk.}
	\label{fig:randomWalk}
\end{figure}

This scenario is represented in Figure \ref{fig:randomWalk}. We recall a corollary of Wald's Identity \cite[Chapter 7.5, Corollary 2]{Gallager:1996}, which is applied to find exponential bounds on the probability of threshold crossing. In particular, the corollary states that if we consider $Q$ as having a finite moment-generating function $\gamma(r)=\ln\{\mathbbm{E}[\exp(rQ)]\}$ over an interval $(r_{-},r_{+})$, a negative drift $\mathbbm{E}[Q]<0$ and $r^*$ being the positive root of $\gamma(r)$, then the probability of crossing threshold $\alpha>0$ by the random walk $S_t=Q_1+Q_2+\ldots +Q_t$ is
\begin{equation}\label{eq:Waldidentity}
\Pr\{S_{\tau} \geq \alpha \} \leq \exp(-r^* \alpha),
\end{equation}
where $\tau$ is the minimum $t$ for which the threshold $\alpha$ is crossed. Having a finite moment generating function means that $Q$ must have moments of all orders and the tails of its distribution function must decay at least exponentially in $q$ as $q \rightarrow \infty$ and $q \rightarrow -\infty$. In our specific setting, $\mathbbm{E}[Q]>0$, $\alpha<0$, and $r^*<0$. We can still apply Wald's identity by changing the signs of $r^*$ and $\alpha$ and by considering the probability of crossing a negative threshold. Thus, we have
\begin{equation}
\Pr\{S_{\tau} \leq \alpha \} \leq \exp(-r^* \alpha),
\end{equation}
where $\alpha<0$ and $r^*<0$. When $\lim_{n \rightarrow \infty} n-s(n) = \infty$ and $\lim_{n \rightarrow \infty} s(n) = \infty$, $\alpha \rightarrow -\infty$ and $\exp(-r^* \alpha) \rightarrow 0$. Thus, we obtain
\begin{equation}
\lim_{n \rightarrow \infty} \Pr\{S_{\tau} \leq \alpha \} = 0.
\end{equation}
\end{proof}

\section{Proof of Lemma \ref{lemma:storeAndHidePPF}}\label{ap:storeAndHidePPF}

\begin{proof}
Split the sequence of input and output symbols into the storage and hiding phases of duration $s(n)$ and $n-s(n)$, respectively and let $s(n) = o(n)$. Then, it is possible to write
\begin{subequations}
\begin{align}
\frac{1}{n} &I(X^n;Y^n) = \frac{1} {n} \Bigg[ \sum_{i=1}^n H(X_i|X^{i-1})-H(X_i|X^{i-1},Y^n)\Bigg] \\ 
&\geq \frac{1} {n} \Bigg[  \sum_{i=1}^n H(X_i) - H(X_i|Y_i), \label{eq:proofSAH2}\Bigg]\\
&= \frac{1}{n} \Bigg[ \sum_{i=1}^{s(n)} I(X_i;Y_i) + \sum_{i=s(n)+1}^{n} I(X_i;Y_i)  \Bigg]\label{eq:proofSAH3}\\
&= \frac{1}{n} \Big\{ s(n) H(X) + [n-s(n)] \mathcal{I}^i(\infty,\hat{P}) \Big\}  \label{eq:proofSAH4}\\
&= \frac{s(n)H(X)}{n} + \frac{[n-s(n)]\mathcal{I}^i(\infty,\hat{P})}{n},  \label{eq:proofSAH5}
\end{align}
\end{subequations}
where (\ref{eq:proofSAH2}) follows because $X$ is i.i.d. and conditioning reduces entropy; (\ref{eq:proofSAH4}) follows since in the first $s(n)$ time instants leakage of full information $H(X)$ takes place, while in the following $n-s(n)$ time slots private operation is assured via the optimal strategy of Theorem \ref{th:Binfinite_peak}.

If we take the limit $n \rightarrow \infty$, since $s(n) = o(n)$ and $H(X)$ is finite, we obtain the leakage rate
\begin{equation}
\lim_{n \rightarrow \infty} \frac{s(n)}{n}H(X) + \frac{n-s(n)}{n}\mathcal{I}^i(\infty,\hat{P}) = \mathcal{I}^i(\infty,\hat{P}).
\end{equation}
\end{proof}

\end{appendices}


\bibliographystyle{IEEEtran}
\bibliography{Ref}

\begin{IEEEbiography}[{\includegraphics[width=1in,height=1.25in,clip,keepaspectratio]{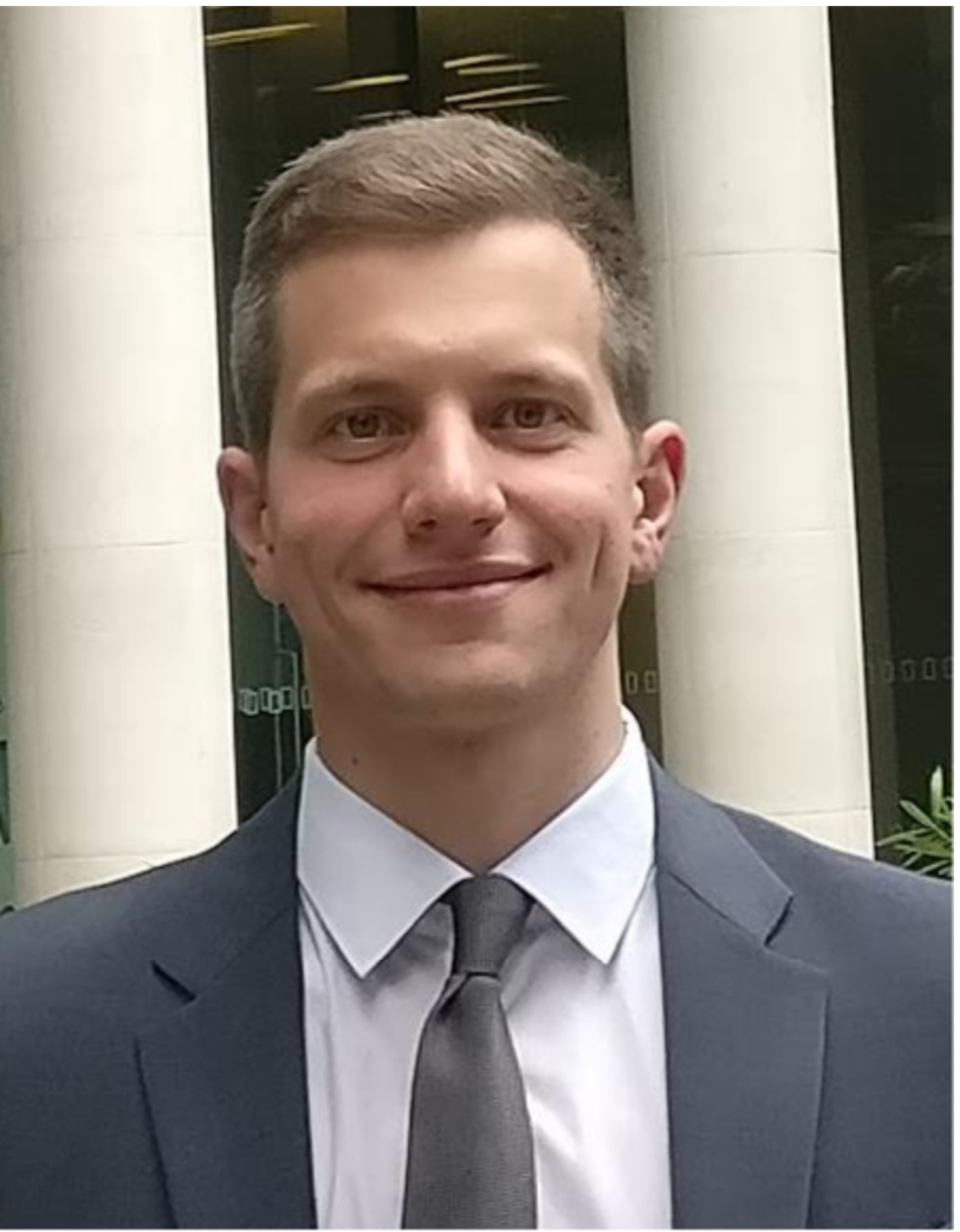}}]
{Giulio Giaconi} (S'15) received the B.Sc. and M.Sc. degrees (Hons.) in communication engineering from the Sapienza University of Rome, Italy, in 2011 and 2013, respectively. He is currently pursuing the Ph.D. degree with the Department of Electrical and Electronic Engineering, Imperial College London, U.K. In 2013, he was a Visiting Student with Imperial College London, working on indoor localization via visible light communications. His current research interests include data privacy, information and communication theory, optimization, signal processing, and machine learning. In 2014, he received the Excellent Graduate Student Award of the Sapienza University of Rome. 
\end{IEEEbiography}

\begin{IEEEbiography}[{\includegraphics[width=1in,height=1.25in,clip,keepaspectratio]{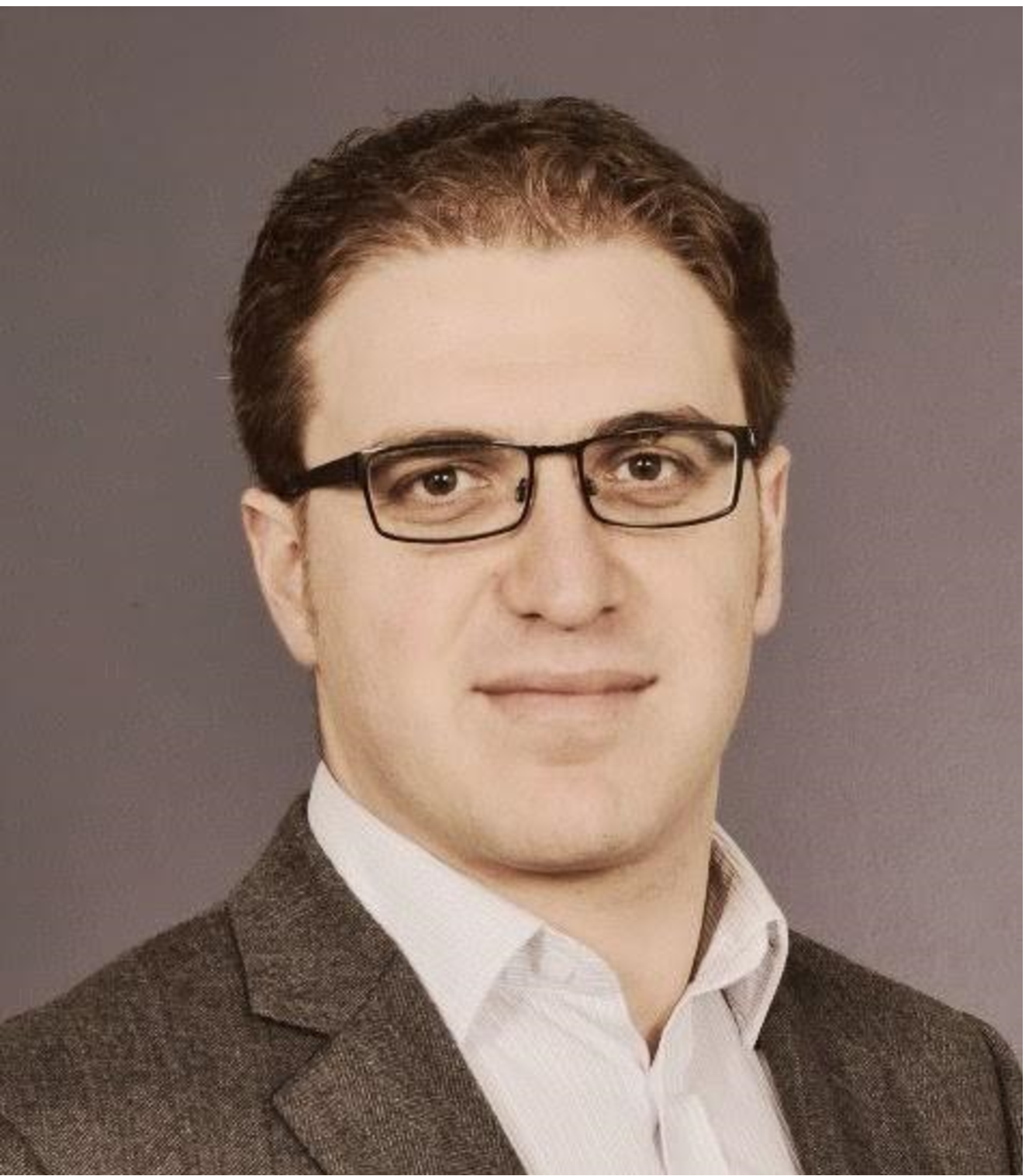}}]%
{Deniz G\"{u}nd\"{u}z} (S'03-M'08-SM'13) received the B.S. degree in electrical and electronics engineering from METU, Turkey in 2002, and the M.S. and Ph.D. degrees in electrical engineering from the NYU Polytechnic School of Engineering in 2004 and 2007, respectively. He served as a Post-Doctoral Research Associate with Princeton University, and a Consulting Assistant Professor with Stanford University. He was a Research Associate with CTTC, Barcelona, Spain, until 2012, when he joined the Electrical and Electronic Engineering Department, Imperial College London, U.K., where he is currently a Reader in information theory and communications.
 
His research interests include the areas of communications and information theory with special emphasis on multi-user communication networks, multimedia content delivery, energy efficient communications and privacy in cyber-physical systems. He was a recipient of a Starting Grant of the European Research Council in 2016, the IEEE Communications Society Best Young Researcher Award for the Europe, Middle East, and Africa Region in 2014, the Best Paper Award at the 2016 IEEE Wireless Communications and Networking Conference, and the Best Student Paper Award at the 2007 IEEE International Symposium on Information Theory. He is the General Co-chair of the 2018 Workshop on Smart Antennas. He served as the General Co-chair of the 2016 IEEE Information Theory Workshop, and a Co-chair of the PHY and Fundamentals Track of the 2017 IEEE Wireless Communications and Networking Conference. 
He is an Editor of the IEEE TRANSACTIONS ON COMMUNICATIONS, and the IEEE TRANSACTIONS ON GREEN COMMUNICATIONS AND NETWORKING. 

\end{IEEEbiography}

\begin{IEEEbiography}[{\includegraphics[width=1in,height=1.25in,clip,keepaspectratio]{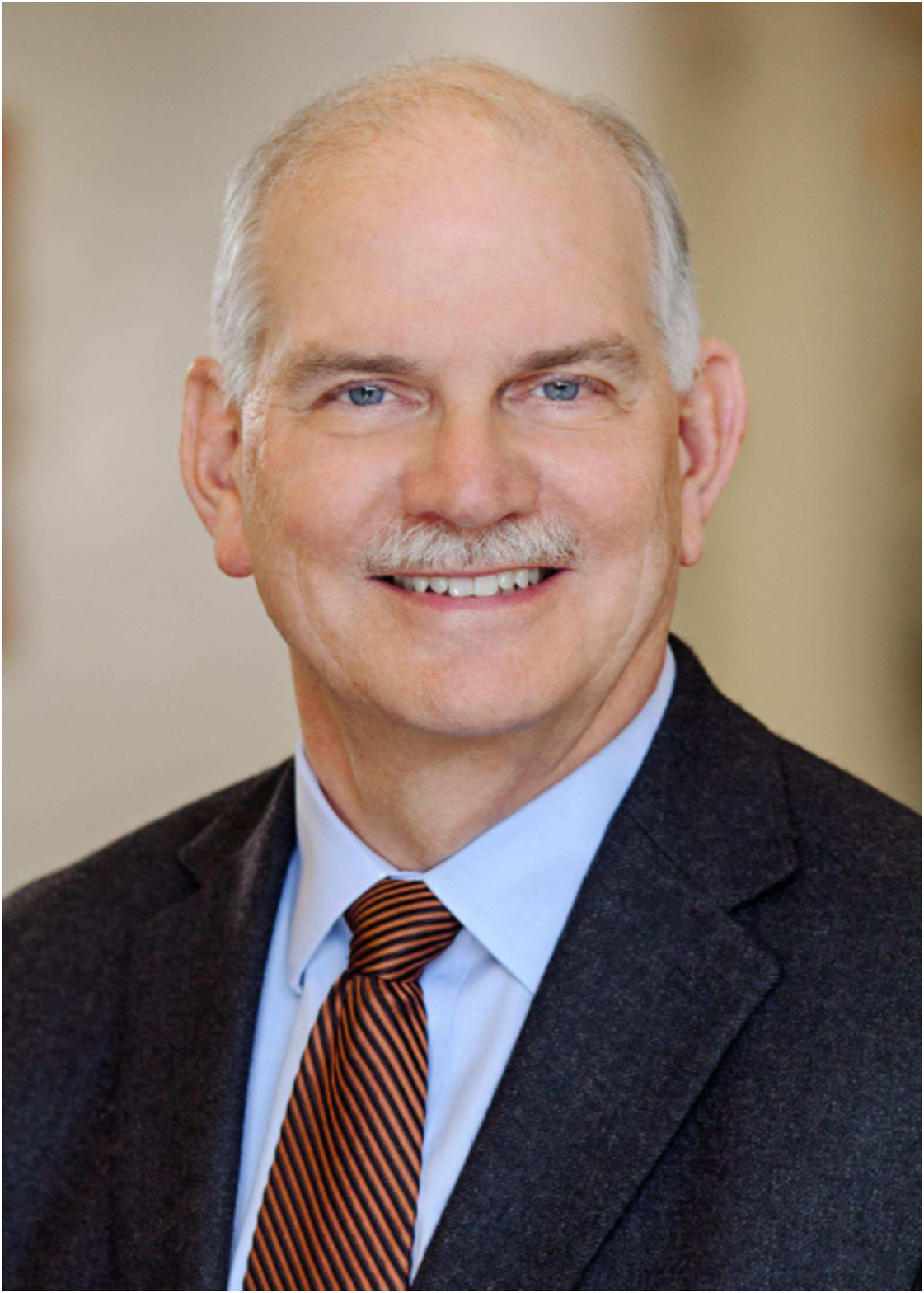}}]%
{H. Vincent Poor} (S'72, M'77, SM'82, F'87) received the Ph.D. degree in EECS from Princeton University in 1977.  From 1977 to 1990, he was with the Faculty of the University of Illinois at Urbana-Champaign. Since 1990 he has been with the Faculty, Princeton, where he is the Michael Henry Strater University Professor of Electrical Engineering. From 2006 to 2016, he served as Dean of Princeton's School of Engineering and Applied Science. He has also held visiting appointments at several other institutions, most recently at Berkeley and Cambridge. His research interests are in the areas of information theory and signal processing, and their applications in wireless networks and related fields. Among his publications in these areas is the recent book \textit{Information Theoretic Security and Privacy of Information Systems} (Cambridge University Press, 2017).

Dr. Poor is a member of the National Academy of Engineering and the National Academy of Sciences, and a foreign member of the Royal Society. In 1990, he served as the President of the IEEE Information Theory Society, and from 2004 to 2007, as an Editor-in-Chief of the IEEE TRANSACTIONS ON INFORMATION THEORY. He received the Technical Achievement and Society Awards of the IEEE Signal Processing Society in 2007 and 2011, respectively. Recent recognition of his work includes the 2017 IEEE Alexander Graham Bell Medal, a D.Sc. \textit{honoris causa} from Syracuse University received in 2017, and election as a Foreign Member of the National Academy of Engineering of Korea in 2017, and an Honorary Member of the National Academy of Sciences of Korea, in 2017.

\end{IEEEbiography}
\end{document}